\documentclass[11pt]{mitthesis}
\usepackage{lgrind}
\usepackage{amsmath, amssymb, amsthm, amsfonts}
\usepackage{fullpage}
\usepackage[pdftex]{graphicx}
\usepackage{caption}
\usepackage{algorithm} 
\usepackage{algorithmic} 
\usepackage{url}
\usepackage{subfigure}

\newtheorem{theorem}{Theorem}
\newtheorem{corollary}[theorem]{Corollary}
\newtheorem{lemma}[theorem]{Lemma}

\newtheorem{proposition}[theorem]{Proposition}
\newtheorem{definition}[theorem]{Definition}
\newtheorem{claim}[theorem]{Claim}

\numberwithin{equation}{section}   

\newcommand{\RR}{\mathbb R}
\newcommand{\ZZ}{\mathbb Z}
\newcommand{\QQ}{\mathbb Q}

\newcommand{\FF}{\mathbb F}
\newcommand{\GGG}{\mathcal G}

\newcommand{\CCC}{\mathcal C}

\newcommand{\FFF}{\mathcal F}
\newcommand{\MMM}{\mathcal M}
\newcommand{\NNN}{\mathcal N}
\newcommand{\III}{\mathcal I}
\newcommand{\AAA}{\mathcal A}
\newcommand{\SSS}{\mathcal S}
\newcommand{\TTT}{\mathcal T}
\newcommand{\PPP}{\mathcal P}
\newcommand{\QQQ}{\mathcal Q}
\newcommand{\OOO}{\mathcal O}

\newcommand{\LLL}{\mathcal L}
\newcommand{\WWW}{\mathcal W}

\newcommand{\closure}{\operatorname{closure}}

\newcommand{\In}{\operatorname{In}}
\newcommand{\Out}{\operatorname{Out}}

\newcommand{\MinCost}{\operatorname{mincost}}

\newcommand{\bd}{\operatorname{bd}}

\begin{document}
%
%
%
%
%
%
%
%

\title{On Network Coding Capacity - Matroidal Networks and Network Capacity Regions}

\author{Anthony Eli Kim}
       \prevdegrees{S.B., Electrical Engineering and Computer Science (2009), and \\
                    S.B., Mathematics (2009) \\
                    Massachusetts Institute of Technology}
\department{Department of Electrical Engineering and Computer Science}

\degree{Master of Engineering in Electrical Engineering and Computer Science}

\degreemonth{September}
\degreeyear{2010}
\thesisdate{August 5, 2010}


\supervisor{Muriel M\'{e}dard}{Professor of Electrical Engineering and Computer Science}

\chairman{Christopher J. Terman}{Chairman, Department Committee on Graduate Theses}

\maketitle



\cleardoublepage
\setcounter{savepage}{\thepage}
\begin{abstractpage}
%
%
%
One fundamental problem in the field of network coding is to determine the network coding capacity of networks under various network coding schemes. In this thesis, we address the problem with two approaches: matroidal networks and capacity regions.

In our matroidal approach, we prove the converse of the theorem which states that, if a network is scalar-linearly solvable then it is a matroidal network associated with a representable matroid over a finite field. As a consequence, we obtain a correspondence between scalar-linearly solvable networks and representable matroids over finite fields in the framework of matroidal networks. We prove a theorem about the scalar-linear solvability of networks and field characteristics. We provide a method for generating scalar-linearly solvable networks that are potentially different from the networks that we already know are scalar-linearly solvable.

In our capacity region approach, we define a multi-dimensional object, called the network capacity region, associated with networks that is analogous to the rate regions in information theory. For the network routing capacity region, we show that the region is a computable rational polytope and provide exact algorithms and approximation heuristics for computing the region. For the network linear coding capacity region, we construct a computable rational polytope, with respect to a given finite field, that inner bounds the linear coding capacity region and provide exact algorithms and approximation heuristics for computing the polytope. The exact algorithms and approximation heuristics we present are not polynomial time schemes and may depend on the output size. 
\end{abstractpage}


\cleardoublepage

\section*{Acknowledgments}
I would like to thank my advisor Muriel M\'{e}dard for her kind support, encouragement, and guidance throughout the course of this work. I am greatly indebted to her not only for introducing me to the field of network coding, but also for providing invaluable advices leading to my intellectual and personal developments. I would also like to thank Una-May O'Reilly for her support and guidance during the early developments leading up to this work and Michel Goemans for thoughtful discussions on parts of the work. I would like to thank members of the Network Coding and Reliable Communications Group of RLE for insightful discussions on network coding problems and for helpful personal advices on many subjects. Finally, I would like to thank my parents, Kihwan Kim and Guiyeum Kim, and sister Caroline for their overflowing love and support throughout my life; and, I would like to thank close friends and members of the MIT community for making my years at MIT unforgettable and, most importantly, enjoyable.

\pagestyle{plain}
\tableofcontents
\newpage
\listoffigures
\newpage
\listoftables
\newpage
\listofalgorithms

\chapter{Introduction}\label{chapter:intro}

Network coding is a field at the intersection of network information theory and coding theory. The central idea of the field of network coding is that increased capabilities of intermediate nodes lead to improvements in information throughput of the network. In the traditional routing model of information networks, intermediate nodes simply copy and forward incoming packets, and the information flows were modeled as source-to-sink paths and Steiner trees. In the network coding model, intermediate nodes are now allowed more complicated operations to code on incoming packets and forward packets that might differ significantly from the incoming packets. It has been shown numerous times that the network coding model allows greater information throughput than in the traditional routing model and its applicability has been widely researched. One fundamental problem in the field of network coding is to determine the network coding capacity, the maximum amount of information throughput, of networks under various network coding schemes. In this work, we address the problem with two approaches: matroidal networks and capacity regions. 


\section{Network Coding Model}
We give a network coding model that we will use in this work. Most of it is adapted from \cite{dougherty:matroid}. Further additional definitions are relegated to relevant chapters. Throughout the work, we assume that the networks are acyclic and the edges (or links) between nodes are delay-free and error-free.

\begin{definition}[Network]
A {\em network} $\NNN$ is a finite, directed, acyclic multigraph given by a 6-tuple $(\nu, \epsilon, \mu, \AAA, S, R)$ where  
\begin{enumerate}
\item $\nu$ is a node set,
\item $\epsilon$ is an edge set,
\item $\mu$ is a message set,
\item $\AAA$ is an alphabet,
\item $S:\nu \rightarrow 2^\mu$ is a source mapping, and
\item $R:\nu \rightarrow 2^\mu$ is a receiver mapping.
\end{enumerate}
\end{definition}

We use a pair of nodes $(x,y)$ to denote a directed edge from node $x$ to node $y$; $x$ is the {\em start node} and $y$ is the {\em end node}. For each node $x$, if $S(x)$ is nonempty then $x$ is a {\em source} and if $R(x)$ is nonempty then $x$ is a {\em receiver}. The elements of $S(x)$ are called the {\em messages generated by $x$} and the elements of $R(x)$ are called the {\em messages demanded by $x$}. An {\em alphabet} $\AAA$ is a finite set with at least two elements. Each instance of a message is a vector of elements from the alphabet. For each node $x$, let $\In(x)$ denote the set of messages generated by $x$ and in-edges of $x$. Let $\Out(x)$ denote the set of messages demanded by $x$ and out-edges of $x$. For each node $x$, we fix an ordering of $\In(x)$ and $\Out(x)$ such that all messages occur before the edges in the resulting lists. In our definition of networks, there could be multiple source nodes and multiple receiver nodes with arbitrary demands. 
	
There are several special classes  of networks: {\em unicast networks} where there are exactly one message, one source and one receiver; {\em multicast networks} where there are exactly one message and one source, but an arbitrary number of receivers that demand the message; {\em two-level multicast networks} where there are multiple messages and there are exactly one source node that generates all the network messages and two receivers where one demands all the messages and the other demands a subset of the messages; {\em multiple unicast networks} where there are multiple messages and, for each message $m$, we have the unicast condition; {\em multiple multicast networks} where there are multiple messages and, for each message $m$, we have the multicast condition. A general network has multiple source nodes, multiple receiver nodes, and arbitrary demands of messages; they are sometimes referred to as {\em multi-source multi-sink networks} in literature. A multicast network in literature usually has multiple messages, but we restrict multicast networks to those with a single message in this work.

We define edge function, decoding function, message assignment and symbol function with respect to a finite field $F$ of cardinality greater than or equal to $|\AAA|$. We choose such $F$ so that each element from $\AAA$ can be uniquely represented with an element from $F$. 

\begin{definition}[Edge and Decoding Functions]
Let $k$ and $n$ be positive integers. For each edge $e=(x,y)$, an {\em edge function} is a map 
\[
f_e: (F^k)^\alpha \times (F^n)^\beta \rightarrow F^n,
\]
where $\alpha$ and $\beta$ are number of messages generated by $x$ and in-edges of $x$, respectively. For each node $x\in \nu$ and message $m\in R(x)$, a {\em decoding function} is a map 
\[
f_{x,m}: (F^k)^\alpha \times (F^n)^\beta \rightarrow F^k,
\]
where $\alpha$ and $\beta$ are number of messages generated by $x$ and in-edges of $x$, respectively. We call $k$ and $n$ the {\em source dimension} and {\em edge dimension}, respectively.
\end{definition}

Each source sends a message vector of length $k$ and each edge carries a message vector of length $n$. We denote the collections of edge and decoding functions by $\FFF_e = \{f_e \: :\: e\in \epsilon\}$ and $\FFF_{d}=\{f_{x,m} \: :\: x\in \nu, m\in R(x)\}$. 

\begin{definition}[Message Assignment]
A {\em message assignment} is a map $a:\mu \rightarrow F^k$, i.e., each message is assigned with a vector from $F^k$. 
\end{definition}

\begin{definition}[Symbol Function]
A {\em symbol function} is a map $s:\epsilon \rightarrow F^n$ defined recursively, with respect to $\NNN$ and $\FFF_e$, such that for all $e=(x,y) \in \epsilon$,
\[
s(e) = f_e(a(m_1), \ldots, a(m_\alpha), s(e_{\alpha+1}), \ldots, s(e_{\alpha+ \beta})), 
\]	
where $m_1, \ldots, m_\alpha$ are the messages generated by $x$ and $e_{\alpha+1}, \ldots, e_{\alpha+\beta}$ are the in-edges of $x$. Note that the symbol function is well-defined as network $\NNN$ is a directed acyclic multigraph. 
\end{definition}

\begin{definition}[Network Code]
A {\em network code} on $\NNN$ is a 5-tuple $(F, k, n,\FFF_e, \FFF_d)$ where 
\begin{enumerate}
\item $F$ is a finite field, with $|F| \geq |\AAA|$,
\item $k$ is a source dimension,
\item $n$ is an edge dimension,
\item $\FFF_e$ is a set of edge functions on network $\NNN$,
\item $\FFF_d$ is a set of decoding functions on network $\NNN$.
\end{enumerate}
\end{definition}

We shall use the prefix $(k,n)$ before codes when we wish to be more specific on parameters $k$ and $n$. When $k$ and $n$ are clear from the context, we will sometimes omit them. There are several special classes of network codes: {\em routing network codes}, where edge and decoding functions simply copy input vector components to output vector components, {\em linear network codes}, where edge and decoding functions are linear over $F$, and {\em nonlinear network codes}, where edge and decoding functions are nonlinear over $F$. {\em Vector-linear network codes} are linear network codes with $k=n$. {\em Scalar-linear network codes} are linear network codes with $k=n=1$. 

\begin{definition}[Network Code Solution]
A network code $(F,k,n,\FFF_e, \FFF_d)$ is a {\em network code solution}, or {\em solution} for short, if for every message assignment $a:\mu \rightarrow F^k$,
\[
f_{x,m}(a(m_1), \ldots, a(m_\alpha), s(e_{\alpha+1}), \ldots, s(e_{\alpha+\beta})) = a(m),
\]
for all $x\in \nu$ and $m\in R(x)$. Note that $m_1, \ldots, m_\alpha$ are messages generated by $x$, and $e_{\alpha+1}, \ldots, e_{\alpha+\beta}$ are in-edges of $x$. If the above equation holds for a particular node $x\in \nu$ and message $m\in R(x)$, then we say {\em node $x$'s demand $m$ is satisfied}.
\end{definition}

A network $\NNN$ is {\em routing-solvable} if it has a routing network code solution. Similarly, we say that network $\NNN$ is {\em linearly solvable} ({\em scalar-linearly solvable, vector-linearly solvable, nonlinearly solvable}) if it has a {\em linear} ({\em scalar-linear, vector-linear, nonlinear}) network code solution.

\section{Previous Works}
In a seminal work in 2000, Ahlswede et al.~\cite{ahlswede:network} introduced the network coding model to the problem of communicating information in networks. They showed that the extended capabilities of intermediate nodes to code on incoming packets give greater information throughput than in the traditional routing model. They also showed that the capacity of any multiple multicast network of a certain class is equal to the minimum of min-cuts between the source node and receiver nodes.

Single source networks and linear network coding are comparatively well-understood. Li et al.~\cite{li:linear} showed that linear network coding is sufficient for certain multiple multicast networks. Koetter and M\'{e}dard \cite{koetter:algebra} reduced the problem of determining scalar-linear solvability to solving a set of polynomial equations over some finite field and suggested connections between scalar-linearly solvable networks and nonempty varieties in algebraic geometry. They showed that scalar-linear solvability of many special case networks, such as two-level multicasts, can be determined by their method. Dougherty et al.~\cite{dougherty:poly} strengthened the connection by demonstrating solvably equivalent pairs of networks and polynomial collections; for any polynomial collection, there exists a network that is scalar-linearly solvable over field $F$ if and only if the polynomial collection is solvable over $F$. It is known that scalar-linear network codes are not sufficient in general. The M-network due to Koetter in \cite{medard:mnetwork} is a network with no scalar-linear solution but has a vector-linear solution. Lehman and Lehman~\cite{lehman} using 3-CNF formulas also provided an example where a vector solution is necessary.

More recently, matroidal approaches to analyze networks have been quite successful. Dougherty et al.~\cite{dougherty:construction,dougherty:matroid} defined and studied matroidal networks and suggested connections between networks and matroids. They used matroidal networks constructed from well-known matroids to show  in \cite{dougherty:insuff} that not all solvable networks have a linear solution over some finite-field alphabet and vector dimension. They also constructed a matroidal network to show that Shannon-type information inequalities are not sufficient for computing network coding capacities in general. Recently, El Rouayheb et al.~\cite{rouayheb:matroid} strengthened the connection between networks and matroids by constructing ``solvably equivalent'' pairs of networks and matroids via index codes with their own construction method; the network has a vector-linear solution over a field if and only if the matroid has a multilinear representation over the same field. In another recent work \cite{sun:networkmatroid}, Sun et al. studied the matroid structure of single-source networks which they define as network matroid and showed connections between the network matroids and a special class of linear network codes.

The capacity regions of networks are less well-understood, but a few explicit outer bounds of capacity regions of networks exist. One easy set of outer bounds is the max-flow/min-cut bounds, which were sufficient in the case of certain multiple multicast networks. Harvey et al.\cite{harvey:capacity} combined information theoretic and graph theoretic techniques to provide a computable outer bound on the network coding capacity regions of networks. Yan et al.\cite{yan:outerbd} gave an explicit outer bound for networks that improved upon the max-flow/min-cut outer bound and showed its connection to a kind of  minimum cost network coding problem. They used their results to compute the capacity region of a special class of 3-layer networks. Thakor et al.\cite{thakor:capacity} gave a new computable outer bound, based on characterizations of all functional dependencies in networks, that is provably tighter than those given in \cite{harvey:capacity} and \cite{yan:outerbd}. 

Recently, explicit characterizations of capacity regions, albeit hard to compute, of networks were given using information theoretic approaches. Yan et al.\cite{yan:capacity} provided an exact characterization of the capacity regions for general multi-source multi-sink networks by bounding the constrained regions in the entropy space. However, they noted that explicitly evaluating the obtained capacity regions remains difficult in general. In a related work, Chan and Grant\cite{chan:capacity} showed that even the explicit characterization of capacity regions for single-source networks can be difficult since the computation of a capacity region reduces to the determination of the nonpolyhedral set of all entropy functions and that linear programming bounds do not suffice.

The routing capacity regions of networks are better understood via linear programming approaches. Cannons et al.\cite{cannons:routing} defined the notion of network routing capacity that is computable with a linear program and showed that every rational number in $(0,1]$ is the routing capacity of some solvable network. Yazdi et al.\cite{yazdi:capacity1, yazdi:capacity2} extended a special case of Farkas Lemma called the ``Japanese Theorem'' to reduce an infinite set of linear constraints to a set of finitely many linear constraints in terms of minimal Steiner trees and applied the results to obtain the routing capacity region of undirected ring networks. In a subsequent work, Kakhbod and Yazdi\cite{kakhbod:routing} provided the complexity results on the description size of the finitely many inequalities obtained in \cite{yazdi:capacity1, yazdi:capacity2} and apply them to the undirected ring networks. 

\section{Our Results}
We organize our contributions into two parts: matroidal networks and network capacity regions. In both approaches, we provide examples to demonstrate our main ideas.

\subsection{Matroidal Networks}
In our matroidal approach, we further study the matroidal networks introduced by Dougherty et al.~\cite{dougherty:matroid}. Our contributions can be summarized as follows and we refer to Chapter~\ref{chapter:matroidal} for details:
\begin{enumerate}
\item We prove the converse of a theorem in \cite{dougherty:matroid} which states that, if a network is scalar-linearly solvable then it is a matroidal network associated with a representable matroid over a finite field.
\item We prove a theorem about the scalar-linear solvability of networks and field characteristics. 
\item We provide a method for generating scalar-linearly solvable networks that are potentially different from the networks that we already know are scalar-linearly solvable.
\end{enumerate}

As a consequence, we obtain a correspondence between scalar-linearly solvable networks and representable matroids over finite fields in the framework of matroidal networks. It also follows that determining scalar-linear solvability of a network $\NNN$ is equivalent to determining the existence of a representable matroid $\MMM$ over a finite field and a valid network-matroid mapping between $\MMM$ and $\NNN$. We obtain a set of scalar-linearly solvable networks that are potentially different from the networks that are already known to be scalar-linearly solvable. 

\subsection{Network Capacity Regions}
In our work concerning the network capacity regions, we continue the research along the lines of work by Cannons et al.~\cite{cannons:routing}. Our contributions can be summarized as follows and we refer to Chapter~\ref{chapter:regions} for details:
\begin{enumerate}
\item We define the network capacity region of networks and prove its notable properties: closedness, boundedness and convexity.
\item We show that the network routing capacity region is a computable rational polytope and provide exact algorithms and approximation heuristics for computing the region.
\item We define the semi-network linear coding capacity region that inner bounds the corresponding network linear coding capacity region, show that it is a computable rational polytope and provide exact algorithms and approximation heuristics for computing it. 
\end{enumerate}

While we present our results for the general directed acyclic networks, they generalize to directed networks with cycles and undirected networks. We note that the algorithms and heuristics we provide do have not polynomial running time in the input size. As our notion of the multi-dimensional network capacity region captures the notion of the single-dimensional network capacity in \cite{cannons:routing}, our present work, in effect, addresses a few open problems proposed by Cannons et al.~\cite{cannons:routing}: whether there exists an efficient algorithm for computing the network routing capacity and whether there exists an algorithm for computing the network linear coding capacity. It follows from our work that there exist combinatorial approximation algorithms for computing the network routing capacity and for computing a lower bound of the network linear coding capacity.

\chapter{Matroidal Networks Associated with Representable Matroids}\label{chapter:matroidal}

In this chapter, we further study the matroidal networks introduced by Dougherty et al.~\cite{dougherty:matroid}. We prove the converse of a theorem in \cite{dougherty:matroid} which states that, if a network is scalar-linearly solvable then it is a matroidal network associated with a representable matroid over a finite field. From \cite{dougherty:matroid} and our present work, it follows that a network is scalar-linearly solvable if and only if it is a matroidal network associated with a representable matroid over a finite field. The main idea of our work is to construct a scalar-linear network code from the network-matroid mapping between the matroid and network. Thereby, we show a correspondence between scalar-linearly solvable networks and representable matroids over finite fields in the framework of matroidal networks. It follows that determining scalar-linear solvability of a network $\NNN$ is equivalent to determining the existence of a representable matroid $\MMM$ over a finite field and a valid network-matroid mapping between $\MMM$ and $\NNN$. We also prove a theorem about the scalar-linear solvability of networks and field characteristics. Using our result and the matroidal network construction method due to Dougherty et al., we note that networks constructed from representable matroids over finite fields are scalar-linearly solvable. The constructed networks are potentially different from the classes of networks that are already known to be scalar-linearly solvable. It is possible that our approach provides a superset, but this is unknown at this time.

\section{Definitions}\label{sec:globallinear}

\begin{definition}[Global Linear Network Code]
A {\em global linear network code} is a 5-tuple $(F$, $k$, $n$, $\phi_{msg}$, $\phi_{edge})$ where 
\begin{enumerate}
\item $F$ is a finite field, with $|F| \geq |\AAA|$, 
\item $k$ is a source dimension, 
\item $n$ is an edge dimension,
\item $\phi_{msg}$ is the global coding vector function on messages, $\phi_{msg}:\mu \rightarrow (F^{k \times k})^{\lvert \mu \rvert}$, such that for message $m$, $\phi_{msg}(m) = (M_1, \ldots, M_{\lvert \mu \rvert})^{\rm T}$ where $M_i$ is a $k \times k$ matrix over $F$, and 
\item $\phi_{edge}$ is the global coding vector function on edges, $\phi_{edge}: \epsilon \rightarrow (F^{n \times k})^{\lvert \mu \rvert}$, such that for each edge $e$, $\phi_{edge}(e) = (M_1, \ldots, M_{\lvert \mu \rvert})^{\rm T}$ where $M_i$ is a $n\times k$ matrix over $F$.
\end{enumerate}
\end{definition}

\begin{definition}[Global Linear Network Code Solution]\label{def:globallinear}
A global linear network code $(F$, $k$, $n$, $\phi_{msg}$, $\phi_{edge})$ is a {\em global linear network code solution}, if $|F| \geq |\AAA|$ and the following conditions are satisfied: 
\begin{enumerate}
\item For each message $m \in \mu$, $\phi_{msg}(m) =$ $ (0, \ldots, 0$,$ I^{k\times k}$, $0, \ldots, 0)^{\rm T}$ where $I^{k\times k}$ is the $k\times k$ identity matrix over $F$ and is in the coordinate corresponding to message $m$.
\item For each node $x \in \nu$ and edge $e\in \Out(x)$, if $\phi_{edge}(e) = (M_1, \ldots, M_{\lvert \mu \rvert})^{\rm T}$, then there exist matrices $C_1, \ldots, C_{\alpha+\beta}$ over $F$ such that $M_i = \sum_{j=1}^{\alpha+\beta} C_j M_i^j$, for $i=1,\ldots, \lvert \mu\rvert$.
\item For each node $x \in \nu$ and message $m\in \Out(x)$, if $\phi_{msg}(m) = (M_1, \ldots, M_{\lvert \mu \rvert})^{\rm T}$, then there exist matrices $C'_1, \ldots, C'_{\alpha+\beta}$ over $F$ such that $M_i = \sum_{j=1}^{\alpha+\beta} C'_j M_i^j$, for $i=1,\ldots, \lvert \mu\rvert$.
\end{enumerate}
Where, if $m_1, \ldots, m_\alpha$ are messages generated by $x$ and $e_{\alpha+1}, \ldots, e_{\alpha+\beta}$ are in-edges of $x$, $\phi_{msg}(m_j) = (M_1^j, \ldots, M_{\lvert \mu \rvert}^j)^{\rm T}$ for $j=1, \ldots ,\alpha$ and $\phi_{edge}(e_j) = (M_1^j, \ldots, M_{\lvert \mu\rvert}^j)^{\rm T}$ for $j=\alpha+1, \ldots ,\alpha+\beta$; $C_1, \ldots, C_\alpha$ are $n\times k$ matrices and $C_{\alpha+1}, \ldots, C_{\alpha+ \beta}$ are $n \times n$ matrices that would appear as coefficients in a linear edge function; and $C'_1, \ldots, C'_\alpha$ are $k\times k$ matrices and $C'_{\alpha+1}, \ldots, C'_{\alpha+ \beta}$ are $k \times n$ matrices that would appear as coefficients in a linear decoding function.
\end{definition}

As with the network codes, we shall sometimes use the prefix $(k,n)$ to emphasize the source and edge dimensions or omit $k$ and $n$ if they are clear from the context. It is straightforward to check that the notions of linear network code solution and global linear network code solution are equivalent, as noted in previous works in algebraic network coding (for instance, \cite{koetter:algebra} for the $k=n=1$ case).  
\begin{proposition}
Let $\NNN = (\nu, \epsilon, \mu, \AAA, S,R)$ be a network. Then, $\NNN$ has a $(k,n)$ linear network code solution if and only if it has a $(k,n)$ global linear network code solution. 
\end{proposition}
\begin{proof}
Let $(F,k,n, \FFF_e, \FFF_d)$ be a $(k,n)$ linear network code solution for $\NNN$. Since $\NNN$ is a directed acyclic graph, we can order nodes in $\nu$ with a topological sort so that each edge go from a lower-ranked node to a higher-ranked node. The ordering of nodes induces an ordering of edges $\hat{e}_1, \ldots, \hat{e}_{|\epsilon|}$ such that no path exists from $\hat{e}_i$ to $\hat{e}_j$ for $i>j$. We define $\phi_{msg}$ for all $m$ and $\phi_{edge}$ for $\hat{e}_1, \ldots, \hat{e}_{|\epsilon|}$ in that order: 

\begin{enumerate}
	\item For each message $m$, we define $\phi_{msg}(m) = (0, \ldots, 0, I^{k\times k}, 0, \ldots, 0)^{\rm T}$ where $I^{k\times k}$ is the $k\times k$ identity matrix and is in the coordinate corresponding to $m$.
	\item For each $\hat{e}_j=(x,y)$, the edge function $f_{\hat{e}_j}$ can be written as 
		\[
		f_{\hat{e}_j}(a(m_1), \ldots, a(m_\alpha), s(e_{\alpha+1}), \ldots, s(e_{\alpha+\beta})) = \sum_{l=1}^\alpha C_l \cdot a(m_l) + \sum_{l=\alpha+1}^{\alpha+\beta} C_l \cdot s(e_l),
		\]
		where $m_1, \ldots, m_\alpha$ are messages generated by $x$ and $e_{\alpha+1}, \ldots, e_{\alpha+\beta}$ are in-edges of $x$; and $C_1, \ldots, C_\alpha$ are $n\times k$ matrices and $C_{\alpha+1}, \ldots, C_{\alpha+\beta}$ are $n\times n$ matrices over $F$. Let $\phi_{msg}(m_j) = (M_1^j , \ldots, M_{\lvert \mu \rvert}^j)^{\rm T}$ for $j=1, \ldots, \alpha$ and $\phi_{edge}(e_j) = (M_1^j , \ldots, M_{\lvert \mu \rvert}^j)^{\rm T}$ for $j=\alpha+1, \ldots, \alpha+\beta$. We define $\phi_{edge}(\hat{e}_j) = (M_1,\ldots, M_{\lvert \mu \rvert})^{\rm T}$ where $M_i = \sum_{l=1}^{\alpha+\beta} C_l \cdot M_i^l$.
	\end{enumerate}

Note that $s(\hat{e}_j) = \sum_{i=1}^{\lvert \mu \rvert} M_i \cdot a(m_i)$. By construction, $(F, k, n, \phi_{msg}, \phi_{edge})$ is a valid $(k,n)$ global linear code that satisfies the first two properties of global linear network code solutions. We check the third property. For each $x\in \nu$ and $m\in R(x)$, the decoding function $f_{x,m}$ can be written as 
\[
f_{x,m}(a(m_1), \ldots, a(m_\alpha), s(e_{\alpha+1}), \ldots, s(e_{\alpha+\beta})) = \sum_{l=1}^\alpha C_l \cdot a(m_l) + \sum_{l=\alpha+1}^{\alpha+\beta} C_l \cdot s(e_l), 
\]
and $f_{x,m}(a(m_1), \ldots, a(m_\alpha), s(e_{\alpha+1}), \ldots, s(e_{\alpha+\beta})) = a(m)$. Note that $m_1, \ldots, m_{\alpha}$ are messages generated at $x$ and $e_{\alpha+1}, \ldots, e_{\alpha+\beta}$ are in-edges of $x$; $C_1, \ldots, C_{\alpha}$ are $k\times k$ matrices and $C_{\alpha+1},\ldots$, $C_{\alpha+\beta}$ are $k\times n$ matrices. Let $\phi_{msg}(m_j) = (M_1^j , \ldots, M_{\lvert \mu \rvert}^j)^{\rm T}$ for $j=1, \ldots, \alpha$ and $\phi_{edge}(e_j) = (M_1^j , \ldots, M_{\lvert \mu \rvert}^j)^{\rm T}$ for $j=\alpha+1, \ldots, \alpha+\beta$. It follows that $[\phi_{msg}(m)]_i = \sum_{l=1}^{\alpha+\beta} C_l \cdot M_i^l$ for all $i$, where $[\phi_{msg}(m)]_i$ denotes the $i$-th coordinate of $\phi_{msg}(m)$; $(F, k,n, \phi_{msg}, \phi_{edge})$ is a $(k,n)$ global linear network code solution.

The converse direction is similar and so we only sketch the proof. Let $(F,k,n,\phi_{msg}, \phi_{edge})$ be a $(k,n)$ global linear network code solution for $\NNN$. For each edge $e$, we define edge function $f_e$ by
\[
f_{e}(a(m_1), \ldots, a(m_\alpha), s(e_{\alpha+1}), \ldots, s(e_{\alpha+\beta})) = \sum_{l=1}^\alpha C_l \cdot a(m_l) + \sum_{l=\alpha+1}^{\alpha+\beta} C_l \cdot s(e_l),
\]
where $C_1, \ldots, C_{\alpha+\beta}$ are some matrices satisfying Definition~\ref{def:globallinear}. For each $x\in \nu$ and $m\in R(x)$, we define decoding function $f_{x,m}$ similarly using matrices $C_1, \ldots, C_{\alpha+\beta}$ from Definition~\ref{def:globallinear}.
\end{proof}

We have the following corollaries from the definitions:
\begin{corollary}
Let $\NNN = (\nu, \epsilon, \mu, \AAA, S,R)$ be a network. Then, $\NNN$ has a $(k,k)$ vector-linear network code solution if and only if it has a $(k,k)$ global vector-linear network code solution.
\end{corollary}
\begin{corollary}
Let $\NNN = (\nu, \epsilon, \mu,\AAA, S,R)$ be a network. Then, $\NNN$ has a scalar-linear network code solution if and only if it has a global scalar-linear network code solution.
\end{corollary}

In this chapter, we will focus on scalar-linear network codes, that is linear network codes with $k=n=1$.
\section{Matroids}

We define matroids and three classes of matroids. See \cite{oxley:matroid} for more background on matroids.

\begin{definition}
A matroid $\MMM$ is an ordered pair $(\SSS, \III)$ consisting of a set $\SSS$ and a collection $\III$ of subsets of $\SSS$  satisfying the following conditions:
	\begin{enumerate}
	\item $\emptyset \in \III$;
	\item If $I \in \III$ and $I'\subseteq I$, then $I'\in \III$;
	\item If $I_1$ and $I_2$ are in $\III$ and $|I_1|<|I_2|$, then there is an element $e$ of $I_2 \setminus I_1$ such that $I_1 \cup \{e\} \in \III$. 
	\end{enumerate}
\end{definition}
The set $\SSS$ is called the {\em ground set} of the matroid $\MMM$. A subset $X$ of $\SSS$ is an {\em independent set} if it is in $\III$; $X$ is a {\em dependent set} if not. A {\em base} $B$ of $\MMM$ is a maximal independent set; for all elements $e\in \SSS \setminus B$, $B \cup \{e\} \notin \III$. It can be shown that all bases have the same cardinality. A {\em circuit} of $\MMM$ is a minimal dependent set; for all elements $e$ in $C$, $C \setminus \{e\} \in \III$. For each matroid, there is an associated function $r$ called {\em rank} that maps the power set $2^\SSS$ into the set of nonnegative integers. The rank of a set $X \subseteq \SSS$ is the maximum cardinality of an independent set contained in $X$. 

\begin{definition}[Matroid Isomorphism]
Two matroids $\MMM_1=(\SSS_1, \III_1)$ and $\MMM_2=(\SSS_2, \III_2)$ are isomorphic if there is a bijection map $\psi$ from $\SSS_1$ to $\SSS_2$ such that for all $X \subseteq \SSS_1$, $X$ is independent in $\MMM_1$ if and only if $\psi(X)$ is independent in $\MMM_2$. 
\end{definition}

\begin{definition}[Uniform Matroids]
Let $c,d$ be nonnegative integers such that $c\leq d$. Let $\SSS$ be a $d$-element set and $\III$ be the collection $\{X \subseteq \SSS \: : \: \lvert X \rvert \leq c\}$. We define the uniform matroid of rank $c$ on the $d$-element set to be $U_{c,d} = (\SSS, \III)$.
\end{definition}

\begin{definition}[Graphic Matroids]
Let $G$ be an undirected graph with the set of edges, $\SSS$. Let $\III=\{X\subseteq \SSS \medspace : \medspace X \text{ does not contain a cycle}\}$. We define the graphic matroid associated with $G$ as $\MMM(G) = (\SSS,\III)$.
\end{definition}

\begin{definition}[Representable/Vector Matroid]
Let $A$ be a $d_1 \times d_2$ matrix over some field $F$. Let $\SSS = \{1,\ldots, d_2\}$ where element $i$ in $\SSS$ corresponds to the $i$th column vector of $A$ and $\III = \{X\subseteq \SSS \medspace :$ $\medspace \text{corresponding}$ $\text{column}$ $\text{vectors}$ $\text{form}$ $\text{an}$ $\text{independent}$ $\text{set}\}$. We define the vector matroid associated with $A$ as $\MMM(A) = (\SSS, \III)$. A matroid $\MMM$ is $F$-representable if it is isomorphic to a vector matroid of some matrix over field $F$. A matroid is representable if it is representable over some field. Note that $F$ is not necessarily finite.
\end{definition}

The bases of $U_{c,d}=(\SSS,\III)$ are exactly subsets of $\SSS$ of cardinality $c$ and the circuits are subsets of $\SSS$ of cardinality $c+1$. Each base of $\MMM(G)$ is a spanning forest of $G$, hence an union of spanning trees in connected components of $G$, and each circuit is a single cycle within a connected component. It is known that the graphic matroids are representable over any field $F$. On the other hand, the uniform matroid $U_{2,4}$ is not representable over $GF(2)$. 

\section{Matroidal Networks}

We define matroidal networks and present a method for constructing matroidal networks from matroids; for more details and relevant results, we refer to \cite{dougherty:matroid}. 
\begin{definition}\label{def:matroidal} 
Let $\NNN$ be a network with message set $\mu$, node set $\nu$, and edge set $\epsilon$. Let $\MMM=(\SSS,\III)$ be a matroid with rank function $r$. The network $\NNN$ is a {\em matroidal network} associated with $\MMM$ if there exists a function $f:\mu \cup \epsilon\rightarrow \SSS$, called the {\em network-matroid mapping}, such that the following conditions are satisfied:
\begin{enumerate}
\item $f$ is one-to-one on $\mu$;
\item $f(\mu)\in \III$;
\item $r(f(\In(x))) = r(f(\In(x)\cup \Out(x)))$, for every $x\in \nu$.
\end{enumerate}
We define $f(A)$ to be $\{f(x) \:|\: x\in A\}$ for a subset $A$ of $\mu\cup \epsilon$.
\end{definition}

\begin{theorem}[Construction Method]\label{thm:construction}
Let $\MMM=(\SSS, \III)$ be a matroid with rank function $r$. Let $\NNN$ denote the network to be constructed, $\mu$ its message set, $\nu$ its node set, and $\epsilon$ its edge set. Then, the following construction method will construct a matroidal network $\NNN$ associated with $\MMM$. We do not address issues of complexity of the method. 

We choose the alphabet $\AAA$ to be any set with at least two elements. The construction will simultaneously construct the network $\NNN$, the network-matroid mapping $f:\mu\cup\epsilon \rightarrow \SSS$, and an auxiliary function $g:\SSS \rightarrow \nu$, where for each $x\in\SSS$, $g(x)$ is either
\begin{enumerate}
\item a source node with message $m$ and $f(m)=x$; or
\item a node with in-degree 1 and whose in-edge $e$ satisfies $f(e)=x$.
\end{enumerate}
The construction is completed in 4 steps and each step can be completed in potentially many different ways:

{\noindent}\textit{\underline{Step 1}}: Choose any base $B=\{b_1, \ldots, b_{r(\SSS)}\}$ of $\MMM$. Create network source nodes $n_1, \ldots, n_{r(\SSS)}$ and corresponding messages $m_1, \ldots, m_{r(\SSS)}$, one at each node. Let $f(m_i)=b_i$ and $g(b_i)=n_i$. 

{\noindent}\textit{\underline{Step 2}}: (to be repeated until no longer possible).\\
Find a circuit $\{x_0, \ldots, x_j\}$ in $\MMM$ such that $g(x_1),\ldots, g(x_j)$ have been already defined but not $g(x_0)$. Then we add:
\begin{enumerate}
\item a new node $y$ and edges $e_1, \ldots, e_j$ such that $e_i$ connects $g(x_i)$ to $y$. Let $f(e_i)=x_i$.
\item a new node $n_0$ with a single in-edge $e_0$ that connects $y$ to $n_0$. Let $f(e_0) =x_0$ and $g(x_0) =n_0$. 
\end{enumerate}

{\noindent}\textit{\underline{Step 3}}: (can be repeated arbitrarily many times).\\
If $\{x_0, \ldots, x_j\}$ is a circuit of $\MMM$ and $g(x_0)$ is a source node with message $m_0$, then add to the network a new receiver node $y$ which demands the message $m_0$ and has in-edges $e_1, \ldots, e_j$ where $e_i$ connects $g(x_i)$ to $y$. Let $f(e_i) =x_i$. 

{\noindent}\textit{\underline{Step 4}}: (can be repeated arbitrarily many times).\\
Choose a base $B=\{x_1, \ldots, x_{r(\SSS)}\}$ of $\MMM$ and create a receiver node $y$ that demands all the network messages and has in-edges $e_1, \ldots, e_{r(\SSS)}$ where $e_i$ connects $g(x_i)$ to $y$. Let $f(e_i)=x_i$. 
\end{theorem}

The following theorem is from \cite{dougherty:matroid}. The original theorem states with a representable matroid, but the same proof still works with a representable matroid over a finite field. 
\begin{theorem}\label{thm:scalarmatroidal}
If a network is scalar-linearly solvable over some finite field, then the network is matroidal. Furthermore, the network is associated with a representable matroid over a finite field.
\end{theorem}
\section{Scalar-linear Solvability}
We prove the converse of Theorem~\ref{thm:scalarmatroidal} and that a network is scalar-linearly solvable over a finite field of characteristic $p$ if and only if the network is a matroidal network associated with a representable matroid over a finite field of characteristic $p$. In what follows, we assume that $d_2 \geq d_1$. 

\begin{lemma}\label{lem:rep}
Let $A$ be a $d_1 \times d_2$ matrix over a finite field $F$ and $\MMM(A)$ be the corresponding representable matroid. Then, there exists an arbitrarily large finite field $F'$ and a $d_1 \times d_2$ matrix $A'$ over $F'$ such that the corresponding matroid $\MMM(A')$ is isomorphic to $\MMM(A)$. 
\end{lemma}
\begin{proof}
We show that any finite field $F'$ that contains $F$ as a subfield works; for instance, extension fields of $F$. We consider the same matrix $A$ over $F'$, so choose $A'=A$, and show that a set of column vectors of $A$ is independent over $F$ if and only if it is independent over $F'$. Assume columns $v_1, \ldots, v_k$ are dependent by some scalars $a_i$'s in $F$, $a_1 v_1 + \cdots + a_k v_k = 0$. Since $F'$ contains $F$, all operations with elements of the subfield $F$ stay in the subfield, and the same scalars still work in $F'$, i.e., $a_1 v_1 + \cdots + a_k v_k = 0$ in $F'$. Hence, the vectors are dependent over $F'$. Assume column vectors $v_1, \ldots, v_k$ are independent over $F$. We extend the set of vectors to a basis of $F^{d_1}$. Then, the matrix formed by the basis has a nonzero determinant over $F$. By similar reasons as before, the same matrix has a nonzero determinant when considered as a matrix over $F'$. Hence, the column vectors of the basis matrix are independent over $F'$ and, in particular, the column vectors $v_1, \ldots, v_k$ are independent over $F'$.
\end{proof}

\begin{theorem}\label{thm:converse}
If a network $\NNN$ is matroidal and is associated with a representable matroid over a finite field $F$, then $\NNN$ is scalar-linearly solvable.
\end{theorem}
\begin{proof}
Let $\NNN =(\nu, \epsilon, \mu, \AAA,S,R)$ be a matroidal network. Let $A$ be the $d_1 \times d_2$ matrix over the finite field $F$ such that $\NNN$ is a matroidal network associated with the corresponding matroid $\MMM(A)=(\SSS,\III)$. By Lemma~\ref{lem:rep}, we assume that the finite field $F$ is large enough to represent all elements in $\AAA$, i.e., $\lvert F \rvert \geq \lvert \AAA \rvert$. By Definition~\ref{def:matroidal}, there exists a network-matroid mapping $f:\mu \cup \epsilon \rightarrow \SSS$. Assume $r(\SSS) = d_1$; otherwise, we remove redundant rows without changing the structure of the matroid. Let $f(\mu)=\{i_1, \ldots, i_{\lvert \mu \rvert}\}$. As $f(\mu)\in \III$, the columns indexed by $f(\mu)$ form an independent set.  We extend $f(\mu)$ to a basis $B$ of $F^{d_1}$, if necessary, by adding column vectors of $A$. Without loss of generality, assume the first $d_1$ columns of $A$ form the basis $B$ after reordering. By performing elementary row operations, we uniquely express $A$ in the form
\begin{equation*}
A = [ I_{d_1}\: | \: A']
\end{equation*}
where $A'$ is a $d_1 \times (d_2-d_1)$ matrix and such that $\{i_1, \ldots, i_{\lvert \mu \rvert}\}$ now corresponds to the first $\lvert \mu \rvert$ columns of $A$. Note that the structure of the corresponding matroid stays the same. We introduce dummy messages $m_{\lvert \mu \rvert +1}, \ldots, m_{d_1}$, if necessary, by adding a disconnected node that generates these messages. We assign global coding vectors on the resulting $\NNN$ as follows:
	\begin{enumerate}
	\item for each edge $e$, let $\phi_{edge}(e) = A_{f(e)}$; and 
	\item for each message $m$, let $\phi_{msg}(m) = A_{f(m)}$,
	\end{enumerate}
where $A_i$ denotes the $i$-th column of $A$. We show that the global linear network code defined above is valid and satisfies all the demands. For each node $x\in \nu$, we have $r(f(\In(x))) = r(f(\In(x) \cup \Out(x)))$. It follows that for each edge $e \in \Out(x)$, $A_{f(e)}$ is a linear combination of $\{ A_{f(e')} \: : \: e'\in \In(x)\}$. Equivalently, $\phi_{edge}(e)$ is a linear combination of coding vectors in $\{\phi_{msg}(m): m\in \In(x)\} \cup \{\phi_{edge}(e) : e\in \In(x)\}$. For each message $m\in \Out(x)$, $A_{f(m)}$ is a linear combination of $\{ A_{f(e')} \: : \: e'\in \In(x)\}$. Similarly, $\phi_{msg}(m)$ is a linear combination of coding vectors in $\{\phi_{msg}(m): m\in \In(x)\} \cup \{\phi_{edge}(e) : e\in \In(x)\}$. Note, furthermore, that $\phi_{msg}(m)$ is the standard basis vector corresponding to $m$. It follows that the global linear network code $(F,\FFF_e, \FFF_d)$ thus defined is a global linear network code solution. Removing the dummy messages, it follows that $\NNN$ is scalar-linearly solvable.
\end{proof}

Given an arbitrary matrix $A$, assigning its column vectors as global coding vectors will not give a global linear network code solution necessarily. In essence, the theorem shows that, while we cannot use column vectors of $A$ directly, we can do the described operations to produce an equivalent representation of $A$ from which we can derive a global linear network code solution. From Theorems \ref{thm:construction} and \ref{thm:converse}, we obtain a method for constructing scalar-linearly solvable networks: pick any representable matroid over a finite field $F$ and construct a matroidal network $\NNN$ using Theorem~\ref{thm:construction}. Combining Theorems~\ref{thm:scalarmatroidal} and \ref{thm:converse}, we obtain the following theorem. 

\begin{theorem}
A network is scalar-linearly solvable if and only if the network is a matroidal network associated with a representable matroid over a finite field.
\end{theorem}

One implication of the theorem is that the class of scalar-linearly solvable networks in the algebraic network coding problem corresponds to the class of representable matroids over finite fields in the framework of matroidal networks. In effect, our results show a connection between scalar-linearly solvable networks, which are tractable networks for network coding, and representable matroids over finite fields, which are also particularly tractable in terms of description size. 

In light of Dougherty et al.'s approach~\cite{dougherty:construction, dougherty:matroid}, relationships between field characteristics and linear solvability of matroidal networks are important. In the case of scalar-linear network codes, we fully characterize a relationship with the following theorem. Note that a network might be a matroidal network with respect to more than one representable matroids of different field characteristics and, thus, is possibly scalar-linearly solvable with respect to fields of different characteristics. 

\begin{theorem}
A network is scalar-linearly solvable over a finite field of characteristic $p$ if and only if the network is a matroidal network associated with a representable matroid over a finite field of characteristic $p$. 
\end{theorem}
\begin{proof}
We extend Theorems \ref{thm:scalarmatroidal} and \ref{thm:converse} and Lemma \ref{lem:rep} to include field characteristic $p$, and the statement follows straightforwardly. 
\end{proof}

\begin{corollary}
Any matroidal network $\NNN$ associated with an uniform matroid is scalar-linearly solvable over a sufficiently large finite field of any characteristic. The same holds for the graphic matroids.
\end{corollary}
\begin{proof}
It is straightforward to show that for any uniform matroid $\MMM$ and a prime $p$, there is a sufficiently large finite field $F$ of characteristic $p$ and a matrix $A$ such that $\MMM$ is a representable matroid associated with $A$ over $F$. The same is true for graphic matroids. 
\end{proof}

As a consequence, any matroidal networks constructed from uniform or graphic matroids will not have interesting properties like those constructed from the Fano and non-Fano matroids in Dougherty et al.~\cite{dougherty:construction, dougherty:matroid}.

\section{Examples}

In this section, we provide examples of scalar-linearly solvable networks that follow from Theorem \ref{thm:converse}. As mentioned before, we get a method for constructing scalar-linearly solvable networks from Theorems \ref{thm:construction} and \ref{thm:converse}: pick any representable matroid over a finite field $F$ and construct a matroidal network. We assume $\AAA=\{0,1\}$ throughout this section.

The Butterfly network $\NNN_1$ in Fig.~\ref{fig:butterfly} is a matroidal network that can be constructed from the uniform matroid $U_{2,3}$. The ground set $\SSS$ of $U_{2,3}$ is $\{a,b,c\}$. Nodes 1-2 are the source nodes and nodes 5-6 are the receiver nodes. See Fig.~\ref{fig:butterfly} and Table~\ref{table:butterfly} for details of the construction and a global scalar-linear network code solution. Note that the sets under `Variables' column are order-sensitive. 

Network $\NNN_2$ in Fig.~\ref{fig:uniform24} is a matroidal network constructed from the uniform matroid $U_{2,4}$. The ground set $\SSS$ of $U_{2,4}$ is $\{a,b,c,d\}$. Nodes 1 and 2 are the source nodes and nodes 7-9 are the receiver nodes. $U_{2,4}$ is a representable matroid associated with $A=\left[\begin{array}{cccc}1& 0 &1&2\\0&1&1& 1\\\end{array}\right]$ over $\FF_3$ and, hence, it has a scalar-linear network code solution over $\FF_3$. See Fig.~\ref{fig:uniform24} and Table~\ref{table:uniform24} for details. 

\begin{figure}[!t]
\parbox[!t]{2.5in}{
\includegraphics[width=2.5in]{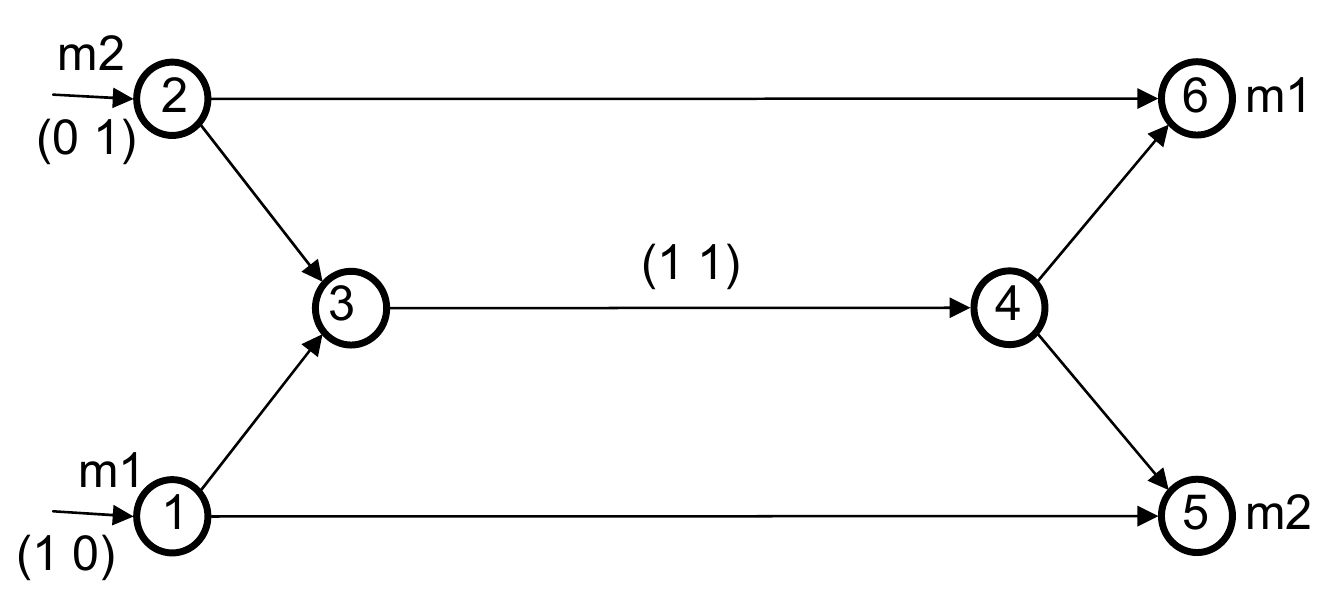}
\captionof{figure}{The Butterfly network $\NNN_1$ constructed from $U_{2,3}$}
\label{fig:butterfly}
}
\hfill
\parbox[!t]{0.45\textwidth}{
\centering
\begin{tabular}{|c|c|c|c|c|c|}
\hline
Step & Variables & Nodes & $x$ &$g(x)$\\
\hline
\hline
1 & $\{a, b\}$& $n_1,n_2$ &  $a$ & $n_1$ \\
  &           &           &  $b$ & $n_2$\\
2 & $\{c,a,b\}$ & $n_3, n_4$ &  $c$ & $n_4$ \\
3 & $\{b,a,c\}$ & $n_5$ &   &  \\
3 & $\{a,b,c\}$ & $n_6$ &  & \\
4 & none & not used & & \\
\hline
\end{tabular}
\captionof{table}{Construction of the Butterfly network $\NNN_1$ from $U_{2,3}$.}
\label{table:butterfly}
}
\end{figure}

%

\begin{figure}[!t]
\parbox[!t]{2.5in}{
\includegraphics[width=2.5in]{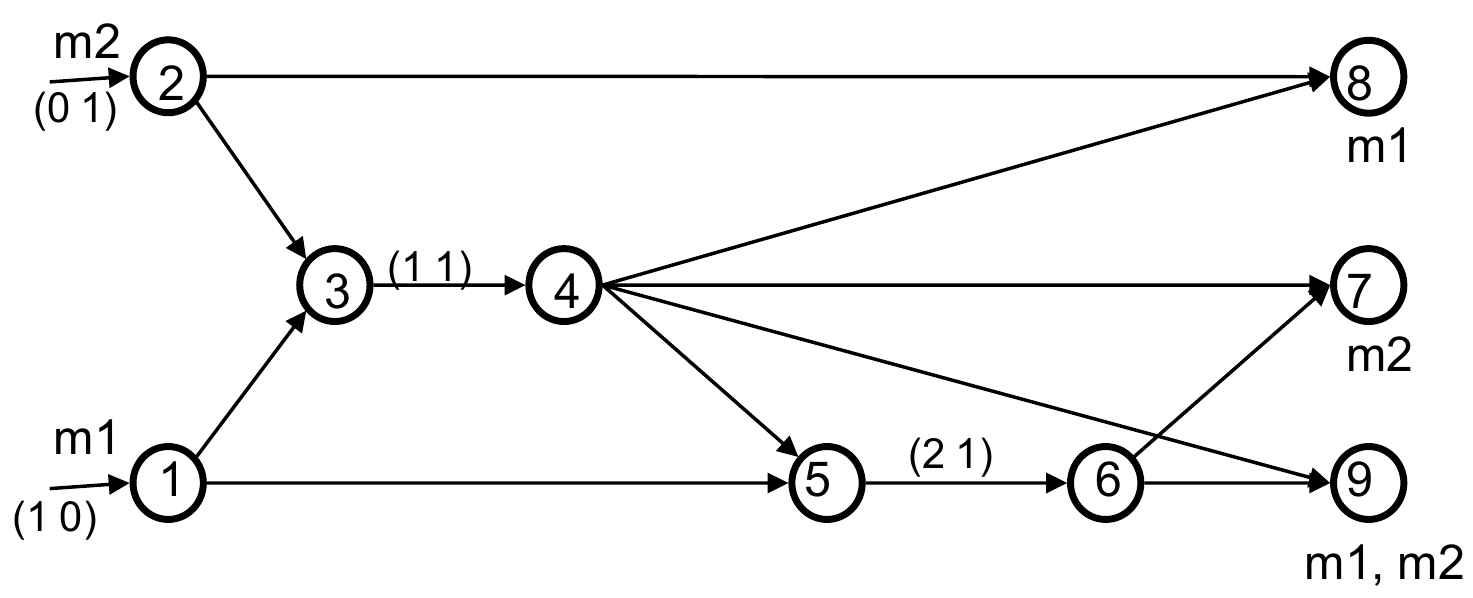}
\captionof{figure}{Network $\NNN_2$ constructed from $U_{2,4}$}
\label{fig:uniform24}
}
\hfill
\parbox[!t]{0.45\textwidth}{
\centering
\begin{tabular}{|c|c|c|c|c|c|}
\hline
Step & Variables & Nodes & $x$ &$g(x)$\\
\hline
\hline
1 & $\{a, b\}$& $n_1,n_2$ &  $a$ & $n_1$ \\
  &           &           &  $b$ & $n_2$\\
2 & $\{c,a,b\}$ & $n_3, n_4$ &  $c$ & $n_4$ \\
2 & $\{d,a,c\}$ & $n_5, n_6$ &  $d$ & $n_6$ \\
3 & $\{b,c,d\}$ & $n_7$ &   &   \\
3 & $\{a,b,c\}$ & $n_8$ &  & \\
4 & $\{c,d\}$ & $n_9$ & & \\
\hline
\end{tabular}
\captionof{table}{Construction of $\NNN_2$ from $U_{2,4}$}
\label{table:uniform24}
}
\end{figure}

%

Consider the graph $G$ and the matroidal network $\NNN_3$ constructed from $\MMM(G)$ in Fig.~\ref{fig:graphic}. The ground set $\SSS$ of $\MMM(G)$ is $\{1,\ldots,7\}$, representing the edges of $G$. Nodes 1-4 are the source nodes and nodes 11-13 are the receiver nodes. $\MMM(G)$ is a representable matroid over field $\FF_2$ and, by Theorem~\ref{thm:converse}, the network has a scalar-linear network code solution over $\FF_2$, as shown by the global coding vectors on $\NNN_3$ in Fig.~\ref{fig:graphic}. This example shows that our results provide networks which are different from the networks previously known to be scalar-linearly solvable such as multicast, 2-level multicast and disjoint multicast networks. It is possible that network $\NNN_3$ can be constructed from a set of polynomials as in Dougherty et al.~\cite{dougherty:poly} or via index codes as in El Rouayheb et al.~\cite{rouayheb:matroid}.

\begin{figure}[!t]
\parbox[!t]{3.0in}{
\includegraphics[width=3.0in]{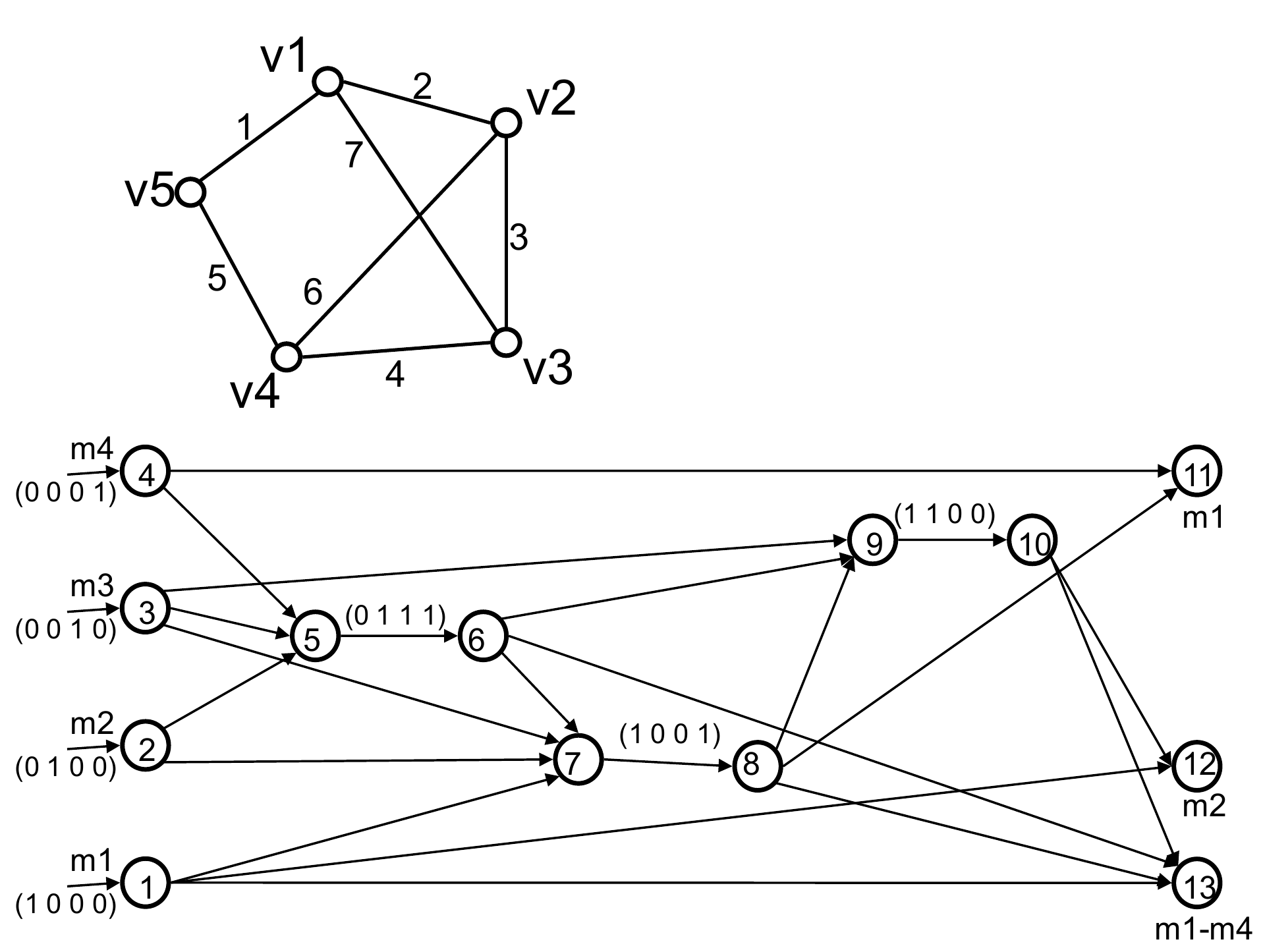}
\captionof{figure}{Graph $G$ and network $\NNN_3$ constructed from $\MMM(G)$}
\label{fig:graphic}
}
\hfill
\parbox[!t]{0.45\textwidth}{
\centering
\begin{tabular}{|c|c|c|c|c|c|}
\hline
Step & Variables & Nodes & $x$ &$g(x)$\\
\hline
\hline
1 & $\{3, 4,5,7\}$& $n_1-n_4$ &  3 & $n_1$ \\
  &               &                     &  4 & $n_2$ \\
  &               &                     &  5 & $n_3$ \\
  &               &                     &  7 & $n_4$ \\
2 & $\{1,4,5,7\}$ & $n_5, n_6$          &  1 & $n_6$\\
2 & $\{2,1,3,4,5\}$ & $n_7, n_8$        &  2 & $n_8$\\
2 & $\{6,1,2,5\}$ & $n_9, n_{10}$       &  6 & $n_{10}$\\
3 & $\{3,2,7\}$   & $n_{11}$            &    &       \\
3 & $\{4,3,6\}$   & $n_{12}$            &    &       \\
4 & $\{1,2,3,6\}$ & $n_{13}$            &    &      \\
\hline
\end{tabular}
\captionof{table}{Construction of $\NNN_3$ from $\MMM(G)$ in Fig.~\ref{fig:graphic}}
\label{table:graphic}
}
\end{figure}

%

\section{Discussion and Conclusion}

In this chapter, we showed that any matroidal network associated with a representable matroid over a finite field is scalar-linearly solvable. Combined with an earlier result of Dougherty et al., it follows that a network is scalar-linearly solvable if and only if it is a matroidal network associated with a representable matroid over a finite field. It also follows that determining scalar-linear solvability of a network is equivalent to finding a representable matroid over a finite field and a valid network-matroid mapping. We also showed a relationship between scalar-linear solvability of networks and field characteristics. Moreover, we obtained a method for generating scalar-linearly solvable networks from representable matroids over finite fields and a set of scalar-linearly solvable networks that are possibly different from those networks that we already know are scalar-linearly solvable.

Unfortunately, the results presented in this chapter do not seem to generalize to vector-linear network coding or more general network coding schemes. The difficulty is that the matroid structure requires that a subset of the ground set of a matroid is either independent or dependent, but what this corresponds to in vector-linear codes, for instance, is not clear. Instead of vectors over fields, we now have vectors over rings (matrices over a field, to be more specific) in vector-linear network coding, and we are unaware of suitable matroids on vectors over rings for our purpose. In fact, El Rouayheb et al.~\cite{rouayheb:matroid} also made a similar observation and suggested that FD-relations are more related to networks than are matroids. 
\chapter{Network Capacity Regions}\label{chapter:regions}

In this chapter, we study the network capacity region of networks along the lines of work by Cannons et al.~\cite{cannons:routing}. We define the network capacity region of networks analogously to the rate regions in information theory and show its notable properties: closedness, boundedness and convexity. In the case of routing, we prove that the network routing capacity region $\CCC_r$ is a computable rational polytope and provide exact algorithms and approximation heuristics for computing the region. In the case of linear network coding, we define an auxiliary region $\CCC'_l$, called the semi-network linear coding capacity region, which is a computable rational polytope that inner bounds the network linear coding capacity region, and provide exact algorithms and approximation heuristics for computing $\CCC'_l$. More specifically, we show the network routing capacity region $\CCC_r$ ($\CCC'_l$) is an image of a higher-dimensional rational polytope under an affine map and consider the computation of $\CCC_r$ ($\CCC'_l$) as the polytope reconstruction problem with a ray oracle. Our results generalize to directed networks with cycles and undirected networks. We note that some underlying problems associated with the polytope reconstruction problem are NP-hard, such as the minimum cost directed Steiner tree problem, and that algorithms and heuristics we present are not polynomial time schemes. Rather, the algorithms and heuristics may have exponential running time in the input size, depending on the intermediate computations and the resulting output's description size. 

As our notion of the multi-dimensional network capacity region captures the notion of the single-dimensional network capacity in \cite{cannons:routing}, our present work, in effect, addresses a few open problems proposed by Cannons et al.~\cite{cannons:routing}: whether there exists an efficient algorithm for computing the network routing capacity and whether there exists an algorithm for computing the network linear coding capacity. Computing the single-dimensional network capacity is equivalent to computing a point on the boundary of the multi-dimensional network capacity region and a ray starting at the origin and reduces to solving associated linear programs in the case of the network routing capacity and (a lower bound of) network linear coding capacity. It follows from our work that there exist combinatorial approximation algorithms for computing the network routing capacity and for computing a lower bound of the network linear coding capacity.

A polytope has two equivalent descriptions; the vertex description in terms of vertices, or the extreme points, of the polytope and the hyperplane description in terms of linear inequalities defining facets of the polytope. In this work, we do not distinguish the two descriptions and use them interchangeably, but we note that converting one description into another can be computationally expensive; we use vertex enumeration algorithms to convert a hyperplane description into a vertex one and facet enumeration algorithms (essentially, convex hull algorithms) for the conversion in the other direction. See \cite{fukuda, handbook, kaibel} for more details on polytopes and relevant algorithms. 

In this chapter, we consider nondegenerate networks where, for each demand of a message at a receiver node, there is a path from a source node generating the message to the receiver node and where no message is both generated and demanded by the same node.
\section{Fractional Network Coding Model}
We define a fractional network coding model. Most definitions are adapted from Cannons et al.\cite{cannons:routing}. We use the notations where we write vectors or points in a multi-dimensional space with a hat as in $\hat{k}$ and let $k_i$ denote the $i$-th coordinate of the vector $\hat{k}$. We also use $[\hat{k}]_i$ and $\hat{k}(i)$ to denote the $i$-th coordinate to avoid confusions when necessary. When it is clear from the context, we omit the hat to avoid cluttering symbols. 

\begin{definition}[Capacitated Network]
A {\em capacitated network} $\NNN$ is a finite, directed, acyclic multigraph given by a 7-tuple $(\nu, \epsilon, \mu, c, \AAA, S, R)$ where
	\begin{enumerate}
		\item $\nu$ is a node set,
		\item $\epsilon$ is an edge set,
		\item $\mu$ is a message set,
		\item $c:\epsilon \rightarrow \ZZ^+$ is an edge capacity function,
		\item $\AAA$ is an alphabet,  
		\item $S:\nu \rightarrow 2^\mu$ is a source mapping, and
		\item $R:\nu \rightarrow 2^\mu$ is a receiver mapping.
	\end{enumerate}
\end{definition}

As we shall use only capacitated networks in this chapter, we use $\NNN$ to denote a capacitated network. We refer to networks defined in Chapter~\ref{chapter:intro} as {\em ordinary networks}. We assume that the messages in $\mu$ are indexed as $m_1,\ldots, m_{|\mu|}$. We define fractional edge function, fractional decoding function, and fractional message assignment with respect to a finite field $F$, where $|F| \geq |\AAA|$, a source dimension vector $\hat{k}$, and an edge dimension $n$:

\begin{definition}[Fractional Edge and Fractional Decoding Functions]
Let $\NNN = (\nu, \epsilon, \mu, c, \AAA, S, R)$ be a capacitated network and $m_1, \ldots, m_{\lvert \mu \rvert}$ be the messages. Let $\hat{k}=(k_1, \ldots, k_{\lvert \mu \rvert})$ be a vector of positive integers and $n$ be a positive integer. For each edge $e=(x,y)$, a {\em fractional edge function} is a map
\[
f_e:(F^{k_{i_1}})\times \cdots \times (F^{k_{i_\alpha}}) \times (F^{nc(e_{\alpha+1})})\times \cdots \times (F^{nc(e_{\alpha+\beta})})\rightarrow F^{nc(e)},
\]
where $m_{i_1}, \ldots, m_{i_\alpha}$ are $\alpha$ messages generated by $x$ and $e_{\alpha+1}, \ldots, e_{\alpha+\beta}$ are $\beta$ in-edges of $x$. For each node $x\in \nu$ and message $m_j\in R(x)$, a {\em fractional decoding function} is a map
\[
f_{x,m_j}: (F^{k_{i_1}})\times \cdots \times (F^{k_{i_\alpha}}) \times (F^{nc(e_{\alpha+1})})\times \cdots \times (F^{nc(e_{\alpha+\beta})}) \rightarrow F^{k_j},
\]
where $m_{i_1}, \ldots, m_{i_\alpha}$ are $\alpha$ messages generated by $x$ and $e_{\alpha+1}, \ldots, e_{\alpha+\beta}$ are $\beta$ in-edges of $x$.

We call $\hat{k} = (k_1, \ldots, k_{\lvert \mu \rvert})$ the {\em source dimension vector}, where $k_i$ is the source dimension for message $m_i$, and $n$ the {\em edge dimension}. We denote the collections of fractional edge and fractional decoding functions by $\FFF_e = \{f_e \: :\: e\in \epsilon\}$ and $\FFF_{x,m}=\{f_{x,m} \: :\: x\in \nu, m\in R(x)\}$, respectively.
\end{definition}

\begin{definition}[Fractional Message Assignment]
Let $\NNN = (\nu, \epsilon, \mu, c, \AAA, S, R)$ be a capacitated network and $m_1, \ldots, m_{\lvert \mu \rvert}$ be the messages. A {\em fractional message assignment} is a collection of maps $\overline{a} = (a_1, \ldots, a_{\lvert \mu \rvert})$ where $a_i$ is a message assignment for $m_i$, $a_i:m_i \rightarrow F^{k_i}$. 
\end{definition}

\begin{definition}[Fractional Network Code]
Let $\NNN=(\nu, \epsilon, \mu, c, \AAA, S, R)$ be a capacitated network and $m_1, \ldots, m_{\lvert \mu \rvert}$ be the messages in $\mu$. A {\em fractional network code} on $\NNN$ is a 5-tuple $(F, \hat{k}, n, \FFF_e, \FFF_d)$ where
	\begin{enumerate}
		\item $F$ is a finite field, with $|F| \geq |\AAA|$, 
		\item $\hat{k} = (k_1, \ldots, k_{\lvert \mu \rvert})$ is a source dimension vector, 
		\item $n$ is an edge dimension,
		\item $\FFF_e$ is a collection of fractional edge functions on $\NNN$,
		\item $\FFF_d$ is a collection of fractional decoding functions on $\NNN$.
	\end{enumerate}
\end{definition}

As with the ordinary network codes in Chapter~\ref{chapter:intro}, we have different kinds of fractional network codes defined analogously: {\em fractional routing network codes}, {\em fractional linear network codes}, and {\em fractional nonlinear network codes}. We shall use the prefix $(\hat{k},n)$ before codes to emphasize the parameters $\hat{k}$ and $n$.

\begin{definition}[Fractional Network Code Solution]
Let $\NNN=(\nu, \epsilon, \mu, c, \AAA, S,R)$ be a capacitated network and $m_1, \ldots, m_{\lvert \mu \rvert}$ be the messages. A fractional network code $(F, \hat{k}, n, \FFF_e, \FFF_d)$ is a {\em fractional network code solution}, or {\em fractional solution} for short, if for every fractional message assignment $\overline{a} = (a_1, \ldots, a_{\lvert \mu\rvert})$,
\[
f_{x,m_j}(a_{i_1}(m_{i_1}),\ldots, a_{i_\alpha}(m_{i_\alpha}), s(e_{\alpha+1}), \ldots, s(e_{\alpha+\beta})) = a_j(m_j),
\]
for all $x\in \nu$ and $m_j\in R(x)$. Note that $m_{i_1}, \ldots, m_{i_\alpha}$ are $\alpha$ messages generated by $x$ and $e_{\alpha+1}, \ldots, e_{\alpha+\beta}$ are $\beta$ in-edges of $x$. If the above equation holds for a particular $x\in \nu$ and message $m\in R(x)$, then we say {\em node $x$'s demand $m$ is satisfied}. 
\end{definition}

As with network code solutions for ordinary networks, we have special classes of fractional network code solutions: {\em fractional routing network code solutions}, {\em fractional linear network code solutions}, and {\em fractional nonlinear network code solutions}. When it is clear from the context, we refer to them by appropriate abridged versions from time to time.

If $(F, \hat{k}, n, \FFF_e, \FFF_d)$ is a fractional network code solution for $\NNN=(\nu, \epsilon, \mu, c, \AAA, S, R)$, source node $x\in \nu$ sends a vector of $k_i$ symbols from $F$ for each message $m_i \in S(x)$; each receiver node $x \in \nu$ demands the original vector of $k_i$ symbols corresponding to message $m_i$ for each $m_i\in R(x)$; and each edge $e$ carries a vector of $c(e)n$ symbols. We refer to coordinates of the symbol vector of length $k_i$ corresponding to message $m_i$ as {\em message $m_i$'s coordinates}. Note that each coordinate of a message is independent from others. We use coordinates and symbols for messages interchangeably; the $i$-th symbol of message $m$ refers to the $i$-th coordinate of the message. We refer to coordinates of the symbol vector of length $c(e)n$ on edge $e$ as {\em edge $e$'s coordinates}. Note that a coordinate of edge $e$ can be {\em active}, meaning it actively carries a symbol in the fractional network code solution, or {\em inactive}, meaning it is not used in the solution. For instance, if an edge with $5$ available coordinates has to send $2$ independent symbols through the edge, then it suffices to use only $2$ coordinates; in this case, the coordinates that carry symbols are active and the other 3 coordinates are inactive. 

We define a notion of {\em minimal network code solutions} as follows:

\begin{definition}[Minimal Fractional Network Code Solution]
A fractional network code solution $(F, \hat{k}, n, \FFF_e, \FFF_d)$ for $\NNN$ is {\em minimal} if the set $A$ of all active coordinates of edges in the solution is minimal, i.e., there exists no $(\hat{k},n)$ fractional network code solution for $\NNN$ with the set of active coordinates that is a strict subset of $A$. 
\end{definition}

\section{Network Capacity Regions}
\subsection{Definitions}
\begin{definition}[Achievable Coding Rate Vector]
Let $\NNN = (\nu, \epsilon, \mu, c, \AAA, S,R)$ be a capacitated network. A vector of positive numbers $\left(\frac{k_1}{n}, \ldots, \frac{k_{\lvert \mu \rvert}}{n}\right) \in \QQ_+^{\lvert \mu \rvert}$ is an {\em achievable coding rate vector} if there exists a fractional network code solution $(F, \hat{k}, n, \FFF_e, \FFF_d)$ for $\NNN$ where $\hat{k} = (k_1, \ldots, k_{\lvert \mu \rvert})$.
\end{definition}

\begin{definition}[Network Capacity Region]
Let $\NNN=(\nu, \epsilon, \mu, c, \AAA, S,R)$ be a capacitated network and $m_1, \ldots, m_{\lvert \mu \rvert}$ be the messages. The {\em network capacity region} $\CCC$ of $\NNN$ is the closure of all achievable coding rate vectors in $\RR^{\lvert \mu \rvert}$,
\[
\CCC = \closure\left\{\frac{\hat{k}}{n} \: :\: \frac{\hat{k}}{n} = \left(\frac{k_1}{n} , \ldots, \frac{k_{\lvert \mu\rvert}}{n}\right) \text{ is an achievable coding rate vectors}\right\}.
\]
By definition, a network capacity region is a set of points in the Euclidean space $\RR_+^{\lvert \mu \rvert}$. 
\end{definition}

There are different classes of achievable coding rate vectors and, hence, corresponding classes of network capacity regions: the {\em network routing capacity region, $\CCC_r$,} which is the closure of all achievable routing rate vectors; the {\em network linear coding capacity region, $\CCC_l$,} which is the closure of all achievable linear coding rate vectors; and the {\em network nonlinear coding capacity region, $\CCC$,} which is the closure of all achievable nonlinear coding rate vectors (or the network capacity region, equivalently).

\subsection{Properties}
We show that the network capacity regions are closed, bounded and convex sets and satisfy an additional property.

\begin{theorem}
Let $\NNN=(\nu, \epsilon, \mu, c, \AAA, S,R)$ be a capacitated network and $m_1, \ldots, m_{\lvert \mu\rvert}$ be the messages. The corresponding network capacity region $\CCC$ is a closed, bounded and convex set in $\RR_+^{|\mu|}$.
\end{theorem}
\begin{proof}
(Closedness) By definition, $\CCC$ is a closure of a set and, hence, closed.

(Boundedness) We show that $\frac{k_i}{n}$ is bounded for all $i$ in the achievable coding rate vector $\left(\frac{k_1}{n}, \ldots, \frac{k_{\lvert \mu \rvert}}{n}\right)$. By symmetry, it suffices to show for $\frac{k_1}{n}$. Let $n$ be the edge dimension, $\nu_1$ be the set of nodes in $\nu$ that generate message $m_1$ and $\gamma$ be the sum of capacities of out-edges of nodes in $\nu_1$. Then, $k_1 \leq \gamma n$ as we cannot send more than $\gamma n$ independent coordinates of message $m_1$ and expect receivers to recover all the information.  Hence, $\frac{k_1}{n} \leq \gamma$. It follows that $\CCC$ is bounded.

(Convexity) Let $x_0, x_1 \in \CCC$ and $\lambda \in [0,1]$. We show that $x= (1-\lambda)x_0 + \lambda x_1 \in \CCC$. We write $x = x_0 + \lambda (x_1 - x_0)$. There exists sequences of achievable coding rate vectors converging to $x_0$ and $x_1$, say $\{y_{0,j}\}$ and $\{y_{1,j}\}$ respectively. Let $\{\lambda_j\}$ be a sequence of rationals converging to $\lambda$. Then, $y_j = y_{0,j} + \lambda_j (y_{1,j} - y_{0,j})$ is an achievable coding rate vector for $j=1,2, \ldots$. Let $\lambda_j = \frac{p}{q}$, $y_{0,j} = \left(\frac{k_1}{n}, \ldots, \frac{k_{\lvert \mu\rvert}}{n}\right)$ and $y_{1,j} = \left(\frac{k'_1}{n'}, \ldots, \frac{k'_{\lvert \mu\rvert}}{n'}\right)$. Then, 
\[
y_j = \left(\frac{(q-p)k_1 n' + p k'_1 n}{q n n'}, \ldots, \frac{(q-p)k_{\lvert \mu \rvert} n' + p k'_{\lvert \mu \rvert} n}{q n n'}\right).
\]
There exists a fractional network code solution $(F, \hat{k}, qnn', \FFF_e, \FFF_d)$ where $\hat{k} = ((q-p)k_1 n'+p k_1' n, \ldots, (q-p)k_{|\mu|} n' + pk'_{|\mu|}n)$; if $NC_1$ and $NC_2$ are two fractional network code solutions with rate vectors $y_{0,j}$ and $y_{1,j}$, then for first $(q-p)nn'$ coordinates we employ $(q-p)n'$ copies of $NC_1$ and for the remaining $pnn'$ coordinates we employ $pn$ copies of $NC_2$. Then,
\begin{align*}
\lvert x - y_j\rvert & = \lvert x_0 + \lambda(x_1 - x_0) - y_{0,j} - \lambda_j(y_{1,j} - y_{0,j}) \rvert \\
		&= \lvert (1-\lambda) x_0  - (1-\lambda_j) y_{0,j} + \lambda x_1 - \lambda_j y_{1,j} \rvert\\
		&\leq \lvert(1-\lambda)x_0 - (1- \lambda_j)y_{0,j}\rvert + \lvert \lambda x_1 - \lambda_j y_{1,j} \rvert.
\end{align*}

Since $\lambda_j \rightarrow \lambda$ and $y_{0,j}\rightarrow x_0$, $(1-\lambda_j)y_{0,j}$ converges to $(1-\lambda)x_0$ and $\lvert(1-\lambda)x_0 - (1- \lambda_j)y_{0,j}\rvert$ can be made arbitrarily small for sufficiently large $j$. Similarly, $\lvert \lambda x_1 - \lambda_j y_{1,j} \rvert$ can be made arbitrarily small for sufficiently large $j$. It follows that the sequence of achievable coding rate vectors $\{y_j\}$ converges to $x$ and that $x\in \CCC$.
\end{proof}

\begin{corollary}
Let $\NNN=(\nu, \epsilon, \mu, c, \AAA, S,R)$ be a capacitated network and $m_1, \ldots, m_{\lvert \mu\rvert}$ be the messages. The corresponding network routing capacity region, $\CCC_r$, and network linear coding capacity region, $\CCC_l$, are closed, bounded and convex regions in $\RR_+^{|\mu|}$.
\end{corollary}

We note that the network capacity regions are of very special kind by definition:
\begin{proposition}
The network capacity region $\CCC$ is a region such that if $\hat{r} = (r_1, \ldots, r_{|\mu|}) \in \CCC$, then the parallelepiped $[0, r_1]\times \ldots \times [0,r_{|\mu|}]$ is contained in $\CCC$. The same holds for the network routing capacity region $\CCC_r$ and the network linear coding capacity region $\CCC_l$.
\end{proposition}

We use $\bd \CCC$ to denote the boundary of the network capacity regions. Similarly, we use $\bd \CCC_r$ and $\bd \CCC_l$ to denote the boundaries of corresponding regions. 

\section{Network Routing Capacity Regions}\label{sec:routing}

In this section, we prove that the network routing capacity region is a bounded rational polytope, and provide exact algorithms and approximation heuristics for computing it. Since multi-source multi-sink networks can be reduced to multiple multicast networks, it suffices to show the results with respect to the multiple multicast networks; for each message $m$, we add a ``super source node'' that generates the message $m$ and connects to source nodes that generate $m$ via edges of infinite, or sufficiently large, capacities. We assume that the given networks in this section are multiple multicast networks for simpler presentation of results.

\subsection{Properties}

\begin{theorem}\label{thm:routingregion}
The network routing capacity region $\CCC_r$ is a bounded rational polytope in $\RR_+^{|\mu|}$ and is computable.
\end{theorem}
\begin{proof}
(Polytope) It suffices to consider minimal fractional routing solutions since any fractional routing solution can be reduced to a minimal one by successively making unnecessary active edge coordinates inactive. For each coordinate of a message $m$, it suffices to route it along a Steiner tree rooted at the source node of $m$ and spanning all the receiver nodes demanding $m$. Hence, any minimal fractional routing solution consists of routing messages along Steiner trees. Let $\TTT_i$ be the set of all Steiner trees rooted at the source node of message $m_i$ and spanning all receiver nodes that demand $m_i$, and $\TTT$ be the union, $\TTT = \TTT_1 \cup \ldots \cup \TTT_{|\mu|}$. Note that $\TTT$ is a finite set. Then, any minimal fractional routing solution $(F, \hat{k}, n, \FFF_e, \FFF_d)$ satisfies the following constraints:
\[
\begin{array}{llll}
\sum_{T\in \TTT} T(e) \cdot x(T) & \leq & c(e) \cdot n, & \forall e\in \epsilon \\
\sum_{T\in \TTT_i} x(T)& = & k_i, & \forall 1\leq i \leq |\mu | \\
x & \geq & 0,
\end{array}
\]
where $x(T)$ is the number of times Steiner tree $T$ is used in the solution and $T(e)$ is an indicator that is 1 if the Steiner tree $T$ uses the edge $e$, or 0 otherwise. Dividing all variables $x(T)$ by $n$, we obtain 
\[
\begin{array}{llll}
\sum_{T\in \TTT} T(e) \cdot x(T) & \leq & c(e), & \forall e\in \epsilon \\
\sum_{T\in \TTT_i} x(T)& = & \frac{k_i}{n}, & \forall 1\leq i \leq |\mu |\\
x & \geq & 0. &
\end{array}
\]
It follows that all minimal fractional routing solutions, after scaling by $n$, satisfy 
\[
\begin{array}{llll}
\sum_{T\in \TTT} T(e) \cdot x(T) & \leq & c(e), & \forall e\in \epsilon \\
x & \geq & 0. &
\end{array}
\]
As the coefficients are in $\QQ$, the above set of inequalities defines a bounded rational polytope $\PPP$, with rational extreme points, in $\RR_+^{|\TTT|}$. The polytope is bounded, because edge capacities are finite and no Steiner tree can be used for routing for infinitely many times. Each minimal fractional routing solution reduces to a rational point inside the polytope $\PPP$, and each rational point $x$ inside $\PPP$ has a minimal fractional routing solution $(F, \hat{k}, n, \FFF_e, \FFF_d)$ that reduces to it, such that 
\[
\left(\frac{k_1}{n}, \ldots, \frac{k_{|\mu|}}{n}\right) = \left(\sum_{T\in\TTT_1}x(T), \ldots, \sum_{T\in\TTT_{|\mu|}}x(T)\right).
\]
To see the latter statement, we take a rational point in $\PPP$, put rationals under a common denominator, and choose appropriate $\hat{k}$ and $n$. As rational points are dense, the closure of the rational points corresponding to minimal fractional routing solutions is exactly $\PPP$. It follows that the network routing capacity region $\CCC_r$ is the image of $\PPP$ under the affine map 
\[
\psi_r: (x(T))_{T\in \TTT} \mapsto \left(\sum_{T\in\TTT_1} x(T), \ldots,\sum_{T\in\TTT_{|\mu|}}x(T)\right).
\]
As the affine map preserves rationality, it follows that the network routing capacity region is a bounded rational polytope in $\RR_+^{|\mu|}$. 

(Computability) We show that we can compute the vertex description (the extreme points) of the polytope $\CCC_r$. We compute the vertices $v_1, \ldots, v_h$ of polytope $\PPP$ by any vertex enumeration algorithm where the starting point can be any point that corresponds to using a single Steiner tree for the maximum number of times allowed by the network. We compute the images of the vertices of $\PPP$ under the affine map $\psi_r$. The network routing capacity region is given by the vertices of the convex hull of points $\psi_r(v_1), \ldots, \psi_r(v_h)$ in $\RR_+^{|\mu|}$.
\end{proof}

The network routing capacity defined by Cannons et al.\cite{cannons:routing} corresponds to a point on the boundary of polytope $\CCC_r$; it is exactly the intersection point between the (outer) boundary $\bd \CCC_r$ and the ray $\hat{x} = (1, \ldots, 1) t, t\geq 0$. As the ray has a rational slope, the intersection point is rational and, hence, Corollary IV.6 in \cite{cannons:routing} follows straightforwardly. We use $\PPP_r$ to denote the ``parent'' polytope, in Theorem~\ref{thm:routingregion}, of the network routing capacity region $\CCC_r$.

\subsection{Algorithms}
We provide exact algorithms and approximation heuristics for computing the network routing capacity region $\CCC_r$. This subsection goes together with next two subsections, so we advise the reader to refer to these subsections as necessary. We assume that a capacitated network $\NNN$ is given if not stated explicitly. We already provided an exact algorithm for computing $\CCC_r$ in the proof of Theorem~\ref{thm:routingregion}, which we refer to as Algorithm \textsc{\tt VertexEnum\_Route}. The algorithm takes the hyperplane description of $\PPP_r$ and outputs the hyperplane description of $\CCC_r$. Since the polytope $\PPP_r$ is defined in a high dimensional space $\RR_+^{|\TTT|}$ where $|\TTT|$ could be exponential in the description size of networks, Algorithm \textsc{\tt VertexEnum\_Route} may not be efficient in practice as there could be exponentially many vertices.

\begin{algorithm}
\caption{Algorithm \textsc{\tt VertexEnum\_Route}($\NNN$)}
\begin{algorithmic}[1]
\STATE Form the hyperplane description of polytope $\PPP_r$.
\STATE Compute the vertices of $\PPP_r$ with a vertex enumeration algorithm and obtain $v_1, \ldots, v_h$.
\STATE Compute the image of the vertices, $\psi_r(v_1), \ldots, \psi_r(v_h)$.
\RETURN convex hull of $\psi_r(v_1), \ldots, \psi_r(v_h)$.
\end{algorithmic}
\end{algorithm}

To design more efficient algorithms and approximation heuristics, we recast the computation of the network routing capacity region as the polytope reconstruction problem with a ray oracle and use related results in literature \cite{cole, gritzmann}. More specifically, we formulate the reconstruction problem as follows:

\begin{definition}[Polytope Reconstruction Problem]
Let $\QQQ$ be a polytope in $\RR^d$ that contains the origin in its interior and $\OOO_{Ray}$ be a ray oracle that given a ray of the form $\hat{x}=\hat{r}t, t\geq 0$, computes the intersection point between the ray and the boundary of $\QQQ$. Compute a polytope description of $\QQQ$ using a finite number of calls to the oracle $\OOO_{Ray}$. 
\end{definition}

We reduce the computation of the network routing capacity region $\CCC_r$ to a polytope reconstruction problem by 1) reflecting $\CCC_r$ around the origin to get a symmetric polytope $\QQQ$ in $\RR^{|\mu|}$ that contains the origin in its interior and 2) solving the linear programs similar to the one in Cannons et al.~\cite{cannons:routing} to implement the ray oracle $\OOO_{Ray}$. To reflect $\CCC_r$, we map all calls to the ray oracle to equivalent calls with rays defined in $\RR_+^{|\mu|}$. We use the algorithm outlined in Section 5 of Gritzmann et al.~\cite{gritzmann} to compute all the facets of the resulting polytope $\QQQ$ and recover the facets of the network routing capacity region $\CCC_r$. We refer to the overall algorithm as Algorithm \textsc{\tt FacetEnum\_Route}. The main idea of the algorithm is to first find a polytope $\QQQ'$ that contains $\QQQ$ and whose facet-defining hyperplanes are a subset of those for $\QQQ$ (Theorem 5.3 in \cite{gritzmann}), and then successively add more facet-defining hyperplanes of $\QQQ$ to $\QQQ'$ by using $\OOO_{Ray}$. In other words, we start with a polytope that contains $\QQQ$ and successively shrink it until it becomes $\QQQ$. By Theorem 5.5 in Gritzmann et al.~\cite{gritzmann}, we need at most 
\[
f_0(\QQQ)+ (|\mu|-1) f_{|\mu|-1}^2(\QQQ) + (5|\mu| - 4) f_{|\mu|-1}(\QQQ)
\]
calls to the ray oracle $\OOO_{Ray}$ to compute the facets, where $f_i(\QQQ)$ denotes the number of $i$-dimensional faces of $\QQQ$ (the $0$-th dimensional faces being the points). Because of the symmetries around the origin, we need at most 
\[
f_0(\CCC_r)+ (|\mu|-1) f_{|\mu|-1}^2(\CCC_r) + (5|\mu| - 4) f_{|\mu|-1}(\CCC_r)
\]
calls to the ray oracle where $f_i(\CCC_r)$ denotes the number of $i$-dimensional faces of $\CCC_r$ that do not contain the origin. 

\begin{algorithm}
\caption{Algorithm \textsc{\tt FacetEnum\_Route}($\NNN$, $\OOO_{Ray}$)}
\begin{algorithmic}[1]
\STATE Form the hyperplane description of $\PPP_r$ in $\RR_+^{|\TTT|}$.
\STATE Internally, reflect $\CCC_r$ around the origin to get the polytope $\QQQ$ in $\RR^{|\mu|}$.
\STATE Using $\OOO_{Ray}$, compute a polytope $\QQQ'$ containing $\QQQ$.
\WHILE {$\QQQ'$ has undetermined facets}
	\STATE Compute the vertices of $\QQQ'$.
	\STATE Using $\OOO_{Ray}$, compute the intersection points on rays defined by the vertices of $\QQQ'$.
	\STATE Add newly found facet-defining hyperplanes of $\QQQ$ to $\QQQ'$.
\ENDWHILE
\STATE Retrieve facets of $\CCC_r$.
\RETURN the facet description of $\CCC_r$.
\end{algorithmic}
\end{algorithm}

Depending on the implementation of $\OOO_{Ray}$, we get exact algorithms and approximation heuristics for computing $\CCC_r$. If we use an exact algorithm for the ray oracle $\OOO_{Ray}$, we get an exact hyperplane description of the network routing capacity region via Algorithm \textsc{\tt FacetEnum\_Route}. If instead we use an approximation algorithm for the oracle that computes some point $r$ such that the actual intersection point lies between $r$ and $Ar$, then we obtain approximation heuristics that compute a set of points $r$ such that the boundary $\bd \CCC_r$ lies between points $r$ and $Ar$. We note that an approximation algorithm for $\OOO_{Ray}$ does not necessary work with Algorithm \textsc{\tt FacetEnum\_Route} to give an approximation algorithm for $\CCC_r$, where an $A$-approximation of $\CCC_r$ would be a polytope $\PPP$ such that $\PPP \subset \CCC_r \subset A \PPP$. While an approximation algorithm for the oracle $\OOO_{Ray}$ does not necessarily lead to a polytope description of $\CCC_r$, it might be faster and more efficient than exact algorithms and, hence, more applicable to compute a quick ``sketch'' of the capacity region $\CCC_r$. One approximation heuristic for computing the region $\CCC_r$ would be to take a sufficiently large number of rays evenly spread apart throughout the space $\RR_+^{|\mu|}$ and use an approximate oracle $\OOO_{Ray}$ to find the approximate intersection points. As there are many simple variations of this approach, we do not go into the details of the heuristics themselves in this work.

\subsection{Implementations of Exact and Approximate Oracle $\OOO_{Ray}$}

In this subsection, we provide both exact and approximation algorithms for the ray oracle $\OOO_{Ray}$ used in Algorithm \textsc{\tt FacetEnum\_Route}. The implementations of the oracle reduce to solving a linear program. We use any linear programming algorithms, such as the ellipsoid algorithm and simplex algorithm, to solve the linear program exactly and obtain an exact oracle $\OOO_{Ray}$. For the approximate ray oracles, we design a combinatorial approximation algorithm using techniques by Garg and K\"{o}nemann~\cite{garg}. Alternatively, we could use the ellipsoid algorithm with an approximate separation oracle, but as the ellipsoid algorithm is slow in practice, this might not be a viable approach. As the network routing capacity region $\CCC_r$ is a rational polytope, it suffices to consider rays with a rational slope in $\QQ_+^{|\mu|}$. 

\subsubsection{Algorithms}\label{subsec:routingalgo}
Given the hyperplane description of the polytope $\PPP_r$,
\[
\begin{array}{llll}
\sum_{T\in \TTT} T(e) \cdot x(T) & \leq & c(e), & \forall e\in \epsilon \\
x & \geq & 0, &
\end{array}
\]
and a ray with a rational slope, $\hat{x} = \hat{q}t, t\geq 0$, we want to compute the rational intersection point of the ray and the boundary of $\CCC_r$. It is straightforward to see that the intersection point is exactly $\lambda_{max}\hat{q}$ where $\lambda_{max}$ is the optimal value to the linear program
\begin{equation}\label{lp:route}
\begin{array}{lllll}
\max & \lambda & & & \\
\operatorname{s.t.}& \sum_{T\in \TTT} T(e)\cdot x(T) & \leq & c(e), & \forall e\in \epsilon \\
 &\sum_{T\in \TTT_i} x(T) & \geq & \lambda q_i, & \forall 1\leq i \leq |\mu| \\
 &x, \lambda & \geq &0. & \\
\end{array}
\end{equation}
Since the coefficients of the linear program are rational, the optimal value $\lambda_{max}$ and the corresponding solution $x$ are also rational. While written in a different form, the linear program is equivalent to the one in Cannons et al.~\cite{cannons:routing}  when the ray is $\hat{x} = (1, \ldots, 1)t, t\geq 0$. We use any linear programming algorithm, such as the ellipsoid algorithm and simplex algorithm, to solve the above linear program exactly and obtain Algorithm \textsc{\tt OracleRayExact\_Route}. We note that the running time of Algorithm \textsc{\tt OracleRayExact\_Route} could be poor as the linear program has exponentially many variables (one for each Steiner tree) and the associated separation problem is NP-hard (the minimum cost directed Steiner tree problem). 

\begin{algorithm}
\caption{Algorithm \textsc{\tt OracleRayExact\_Route}($\NNN$, $\hat{q}$)}
\begin{algorithmic}[1]
\STATE Form the linear program~\eqref{lp:route} corresponding to $\NNN$ and $\hat{q}$.
\STATE Solve the linear program with a linear programming algorithm and obtain $\lambda_{max}$.
\RETURN $\lambda_{max}\hat{q}$.
\end{algorithmic}
\end{algorithm}

We now provide a combinatorial approximation algorithm for solving the linear program~\eqref{lp:route} and, hence, for oracle $\OOO_{Ray}$. It computes a point $\hat{r}$ such that $\lambda_{max}\hat{q}$ is on the line segment between $\hat{r}$ and $(1+\omega)A\hat{r}$ for some numbers $\omega>0$ and $A\geq 1$. The main idea is to view solving \eqref{lp:route} as concurrently packing Steiner trees according to the ratio defined by $\hat{q}$, and use the results for the multicommodity flow and related problems by Garg and K\"{o}nemann~\cite{garg}. Instead of a shortest path algorithm, we use a minimum cost directed Steiner tree algorithm for our purpose. We assume we have an oracle $\OOO_{DSteiner}$ that solves the minimum cost directed Steiner tree problem, which is well-known to be NP-hard, within an approximation guarantee $A$:

\begin{definition}[Minimum Cost Directed Steiner Tree Problem]\label{prob:DSteiner}
Given an acyclic directed multigraph $\GGG = (\nu, \epsilon)$, a length function $l:\epsilon \rightarrow \RR^+$, a source node $s$ and receiver nodes $n_1, \ldots, n_k$, find a minimum cost subset of edges $\epsilon'$ such that there is a directed path from $s$ to each $n_i$ in $\epsilon'$. The cost of a subset $\epsilon'$ is $\sum_{e\in \epsilon'} l(e)$. 
\end{definition}

First, we consider networks with exactly one message to route. In this case, the linear program~\eqref{lp:route} reduces to the following simpler linear program:
\begin{equation}\label{lp:singlepack}
\begin{array}{lllll}
\max & \sum_{T\in \TTT} x(T)& & & \\
\operatorname{s.t.}& \sum_{T\in \TTT} T(e)\cdot x(T) & \leq & c(e), & \forall e\in \epsilon \\
 &x & \geq &0, & \\
\end{array}
\end{equation}
where $\TTT$ is the set of all Steiner trees rooted at the source node of the message and spanning all the receiver nodes demanding the message. Note that the original problem now reduces to the fractional directed Steiner tree packing problem (compare to the fractional Steiner tree packing problem in \cite{jain}). To solve \eqref{lp:singlepack}, we use Algorithm \textsc{\tt DSteinerTreePacking} which is a straightforward modification of the maximum multicommodity flow algorithm given in Section 2 of Garg and K\"{o}nemann~\cite{garg}. In Algorithm \textsc{\tt DSteinerTreePacking}, $\MinCost(l)$ denotes the cost of the (approximate) minimum cost directed Steiner tree found by $\OOO_{DSteiner}$.

\begin{algorithm}
\caption{Algorithm \textsc{\tt DSteinerTreePacking}($\NNN$, $\omega$, $O_{DSteiner}$, $A$)}
\begin{algorithmic}[1]
\STATE $\eta = \frac{3}{16}\omega$; $\delta = (1+\eta)((1+\eta)L)^{-1/\eta}$
\STATE $f=0$; $l(e) =\delta, \forall e \in \epsilon$
\WHILE {$\MinCost(l) < A$}
	\STATE Use $\OOO_{DSteiner}$ to compute an approximate minimum cost Steiner tree $\widetilde{T}$, under the length $l$.
	\STATE $d = \min_{e: \widetilde{T}(e) = 1} c(e)$
	\STATE $f = f + d$
	\STATE Update $l(e) = l(e)\left(1+\eta \frac{d}{c(e)}\right), \forall e$ s.t. $\widetilde{T}(e) = 1$.
\ENDWHILE
\RETURN $f/ \log_{1+\eta} \frac{1+\eta}{\delta}$ 
\end{algorithmic}
\end{algorithm}
Essentially the same analysis in \cite{garg} works for Algorithm \textsc{\tt DSteinerTreePacking} except that our approximation guarantee is worse by a factor of $A$ since we use an approximate oracle $\OOO_{DSteiner}$. We omit the analysis and summarize the performance of the algorithm as follows. For computations involving $\eta$ and $\omega$, we refer to Appendix~\ref{app:comp1}.

\begin{theorem}\label{thm:singlepack}
For $0< \omega <1$, Algorithm \textsc{\tt DSteinerTreePacking} computes a $(1+\omega)A$-approximate solution to the linear program~\eqref{lp:singlepack} in time $O(\omega^{-2} |\epsilon| \log L \cdot T_{DSteiner})$, where $L$ is the maximum number of edges used in any Steiner tree in $\TTT$ and $T_{DSteiner}$ is the time required by oracle $\OOO_{DSteiner}$ to solve the minimum cost directed Steiner tree problem within an approximation guarantee of $A\geq 1$. Note that $L \leq |\epsilon|$. 
\end{theorem}

We now give Algorithm \textsc{\tt OracleRayApprox\_Route}, for oracle $\OOO_{Ray}$, which computes a $(1+\omega)A$-approximate solution to the linear program~\eqref{lp:route}. It uses Algorithm \textsc{\tt DSteinerTreePacking} as a subroutine. 

\begin{algorithm}
\caption{Algorithm \textsc{\tt OracleRayApprox\_Route}($\NNN$, $\omega$, $\OOO_{DSteiner}$, $A$, $\hat{q}$)}
\begin{algorithmic}[1]
\STATE Using Algorithm \textsc{\tt DSteinerTreePacking}, compute the approximate value $z_i$ to the linear program~\eqref{lp:singlepack} for each message $m_i$ separately. 
\STATE $z = \min_i \frac{z_i}{q_i}$; $scaling\_factor = \frac{|\mu|}{z}$
\STATE $\hat{q} = scaling\_factor \cdot \hat{q}$
\STATE $\eta = \frac{1}{9A} \omega$; $\delta = (|\epsilon|/(1-\eta A))^{-1/\eta A}$
\STATE $N = 2\left\lceil\frac{1}{\eta A} \log_{1+\eta} \frac{|\epsilon|}{1-\eta A} \right\rceil$
\WHILE {\TRUE}
	\STATE $t=0$; $l(e) = \frac{\delta}{c(e)}, \forall e\in \epsilon$
	\FOR {phase $i=1, \ldots$, $N$}
		\FOR {iteration $j=1, \ldots, |\mu|$}
			\STATE $\gamma = q_j$
			\WHILE {$\gamma > 0$}
				\STATE Using $\OOO_{DSteiner}$, compute the minimum cost directed Steiner tree $\widetilde{T}\in T_j$, under $l$.
				\STATE $d = \min \{\gamma, c(e):\widetilde{T}(e) = 1\}$
				\STATE $\gamma = \gamma -d$
				\STATE Update $l(e) = l(e) \left(1+ \eta \frac{d}{c(e)}\right), \forall e \textrm{ such that } \widetilde{T}(e) = 1$.
				\IF {$\sum_e l(e)c(e) \geq 1$} 
					\STATE Goto Line~\ref{lastline}.
				\ENDIF
			\ENDWHILE
		\ENDFOR
		\STATE $t= t+1$
	\ENDFOR
	\STATE $scaling\_factor = 2 scaling\_factor$; $\hat{q} = 2 \hat{q}$
\ENDWHILE
\RETURN $scaling\_factor \cdot t / \log_{1+\eta} \frac{1}{\delta}$ \label{lastline}
\end{algorithmic}
\end{algorithm}

\subsubsection{Analysis}
While Algorithm \textsc{\tt OracleRayApprox\_Route} is closely related to the maximum concurrent flow algorithm given in Section 5 of Garg and K\"{o}nemann~\cite{garg}, there are significant differences and we give an analysis of the algorithm below. The linear program~\eqref{lp:route} can be thought of as the ``concurrent'' fractional directed Steiner tree packing problem. We have directed Steiner trees partitioned into different groups according to the network messages. Then, computing the optimal value of the linear program~\eqref{lp:route} is equivalent to fractionally routing along the Steiner trees so that the ratios among the overall usages of groups of Steiner trees correspond to the ratios among the coordinates of the $\hat{q}$ vector.

The dual linear program corresponding to \eqref{lp:route} is 
\begin{equation}\label{lp:routedual}
\begin{array}{lllll}
\min & \sum_e l(e)c(e)& & & \\
\operatorname{s.t.}& \sum_{e:T(e) = 1} l(e) - z(i) & \leq & 0, & \forall 1\leq i \leq |\mu|, \forall T\in \TTT_i \\
 & \sum_{i=1}^{|\mu|} q_i z(i) & \geq & 1 & \\ 
 &l, z & \geq &0. & \\
\end{array}
\end{equation}
From the dual linear program, we note that the associated separation problem is exactly the minimum cost directed Steiner tree problem. Let $\MinCost_i(l)$ denote the cost of the minimum cost directed Steiner tree in $\TTT_i$ under the length function $l$. Define $D(l) = \sum_{e\in \epsilon} l(e) c(e)$ and $\alpha(l) = \sum_{i=1}^{|\mu|} q_i \MinCost_i(l)$. Then, solving the dual linear program is equivalent to finding an assignment of lengths to edges, $l:\epsilon \rightarrow \RR^+$, so as to minimize $\frac{D(l)}{\alpha(l)}$. Let $\beta$ be the optimal value of the dual linear program, i.e., $\beta = \min_l \frac{D(l)}{\alpha(l)}$. 

We first consider a modified version of Algorithm \textsc{\tt OracleRayApprox\_Route}, with the infinite while-loop in line 6 removed, with variable $scaling\_factor$ and line 23 removed, and with the finite for-loop in line 8 replaced with an infinite while-loop. The following holds:

\begin{theorem}\label{routing:step1}
Assume $\beta\geq 1$. For $0< \omega <1$, Algorithm \textsc{\tt OracleRayApprox\_Route}, with the modifications, returns a $(1+\omega)A$-approximate solution to the linear program~\eqref{lp:route} in at most $\left\lceil \frac{\beta}{\eta A} \log_{1+\eta} \frac{|\epsilon|}{1-\eta A}\right\rceil$ number of phases.
\end{theorem}
\begin{proof}
Algorithm \textsc{\tt OracleRayApprox\_Route} in lines 8-22 proceeds in phases which in turn consist of $|\mu|$ iterations. Each iteration consists of variably many number of steps of the while-loop in line 11, depending on how quickly $\gamma$ is decreasing. In each $j$-th iteration in line 9, $q_j$ units of message $m_j$ are routed using Steiner trees in $\TTT_j$. By lines 16-18, the algorithm terminates as soon as $D(l)\geq 1$. Let $t$ be the last phase in which the algorithm terminates. Let $l_{i,j}^{s-1}$ denote the length function $l$ and $\gamma_{i,j}^{s-1}$ the variable $\gamma$ at the start of $s$-th step of the while-loop (line 11) of $j$-th iteration in phase $i$. Let $\widetilde{T}_{i,j}^{s}$ denote the Steiner tree $\widetilde{T}$ selected in the $s$-th step of $j$-th iteration in phase $i$. Let $l_{i,0}$ be the length function at the start of phase $i$ and $l_{i,|\mu|}$ be the length function at the end of phase $i$ (after the termination of the $|\mu|$-th iteration). Note that $l_{i+1, 0} = l_{i, |\mu|}$. Let $l_{i,j-1}$ be the length function at the start of iteration $j$ of phase $i$. For simplicity, we denote $D(l_{i,|\mu|})$ and $\alpha(l_{i,|\mu|})$ by $D(i)$ and $\alpha(i)$, respectively. Then, 
\begin{align*}
D(l_{i,j}^s) &= \sum_e l_{i,j}^s(e) \cdot c(e)\\
 &= D(l_{i,j}^{s-1}) + \eta(\gamma_{i,j}^{s-1} - \gamma_{i,j}^s) \sum_{e:\widetilde{T}_{i,j}^{s}(e)=1} l_{i,j}^{s-1}(e) \\
 &\leq D(l_{i,j}^{s-1}) + \eta(\gamma_{i,j}^{s-1} -\gamma_{i,j}^s) A \MinCost_j(l_{i,j}^{s-1}).
\end{align*}
Then,
\begin{align*}
D(l_{i,j+1}) &\leq D(l_{i,j}) + \eta A \sum_s (\gamma_{i,j}^{s-1}-\gamma_{i,j}^s) \MinCost_{j+1}(l_{i,j+1})\\
 &= D(l_{i,j}) + \eta A q_{j+1} \MinCost_{j+1}(l_{i,j+1})\\
 &\leq D(l_{i,j}) + \eta A q_{j+1} \MinCost_{j+1}(l_{i,|\mu|}),
\end{align*}
where we used the fact that $\MinCost_j(l_{i,j}^s) \leq \MinCost_j(l_{i,j+1})$ as the length function and $\MinCost$ are nondecreasing in any fixed argument throughout the algorithm. After summing up the above inequalities over the iterations of phase $i$, we get
\begin{align*}
D(l_{i, |\mu|}) &\leq  D(l_{i,|\mu|-1}) + \eta A q_{|\mu|} \MinCost_{|\mu|} (l_{i,|\mu|})\\
 &\leq D(l_{i,|\mu|-2}) + \eta A(q_{|\mu|-1}\MinCost_{|\mu|-1}(l_{i,|\mu|}) + q_{|\mu|}\MinCost_{|\mu|}(l_{i,|\mu|}))\\
 &\phantom{XXXXXXXXXXX} \vdots \\
 &\leq D(l_{i,0}) + \eta A \sum_{j=1}^{|\mu|} q_j \MinCost_j (l_{i,|\mu|}),
\end{align*}
and it follows that 
\[
D(i) \leq D(i-1) + \eta A \alpha(i).
\]

Since $\frac{D(i)}{\alpha(i)} \geq \beta$, it follows that $D(i) \leq \frac{D(i-1)}{1-\eta A/\beta}$. Since $D(0) = |\epsilon|\delta$, we have for $i\geq 1$, 
\begin{align*}
D(i) &\leq \frac{|\epsilon|\delta}{(1-\eta A/\beta)^i}\\
  &= \frac{|\epsilon|\delta}{1-\eta A /\beta} \left(1+ \frac{\eta A}{\beta - \eta A} \right)^{i-1} \\
  &\leq \frac{|\epsilon|\delta}{1-\eta A/\beta} e^{\frac{\eta A(i-1)}{\beta-\eta A}}\\
  &\leq \frac{|\epsilon| \delta}{1-\eta A} e^{\frac{\eta A(i-1)}{\beta(1-\eta A)}},
\end{align*}
where the last inequality uses the assumption that $\beta\geq 1$. The algorithm terminates in phase $t$ for which $D(t)\geq 1$. Hence,
\[
1 \leq D(t) \leq \frac{|\epsilon|\delta}{1-\eta A} e^\frac{\eta A(t-1)}{\beta(1- \eta A)}.
\]
And it follows that 
\begin{equation}\label{ratio1}
\frac{\beta}{t-1} \leq \frac{\eta A}{(1-\eta A) \ln \frac{1-\eta A}{|\epsilon|\delta}}.
\end{equation}
In the first $t-1$ phases, we have routed $(t-1)q_j$ units of message $m_j$, for $j=1, \ldots, |\mu|$. This routing solution may violate the edge capacity constraints, but, by the following claim, we obtain a feasible solution with $\lambda > \frac{t-1}{\log_{1+\eta} 1/\delta}$. See Claim 5.1 in \cite{garg} for the proof of the claim as the same proof works here.

\begin{claim}
There exists a feasible solution with $\lambda > \frac{t-1}{\log_{1+\eta} \frac{1}{\delta}}$.
\end{claim}

Thus, the ratio of the values of the optimal dual and feasible primal solutions, $\zeta$, is strictly less than $\frac{\beta}{t-1} \log_{1+\eta} \frac{1}{\delta}$. By \eqref{ratio1}, we get 
\[
\zeta < \frac{\eta A}{1-\eta A} \frac{\log_{1+\eta} \frac{1}{\delta}}{\ln \frac{1-\eta A}{|\epsilon|\delta}} = \frac{\eta A}{(1-\eta A) \ln (1+\eta)} \frac{\ln 1/\delta}{\ln \frac{1-\eta A}{|\epsilon| \delta}}.
\]
For $\delta = (|\epsilon|/(1-\eta A))^{-1/\eta A}$, the ratio $\frac{\ln 1/\delta}{\ln \frac{1-\eta A}{|\epsilon| \delta}}$ equals $(1-\eta A)^{-1}$ and, hence,
\[
\zeta \leq \frac{\eta A}{(1-\eta A)^2 \ln(1+\eta)} \leq \frac{\eta A}{(1-\eta A)^2 (\eta - \eta^2/2)} \leq (1-\eta A)^{-3}A.
\]
For $\eta = \frac{1}{9A} \omega$, $(1-\eta A)^{-3} A$ is at most our desired approximation ratio $(1+\omega) A$ (see Appendix~\ref{app:comp2} for details). By weak-duality, we have
\[
1 \leq \zeta < \frac{\beta}{t-1}\log_{1+\eta} \frac{1}{\delta}
\]
and, therefore, the number of phases in line 8 is strictly less than $1+\beta \log_{1+\eta} 1/\delta$, which implies that Algorithm \textsc{\tt OracleRayApprox\_Route} terminates in at most $\left\lceil \frac{\beta}{\eta A} \log_{1+\eta} \frac{|\epsilon|}{1-\eta A}\right\rceil$ number of phases. 
\end{proof}

Note that by Theorem~\ref{routing:step1}, the running time of Algorithm \textsc{\tt OracleRayApprox\_Route}, with the modifications, depends on $\beta$. Note that $\beta$ can be reduced/increased by scaling the $\hat{q}$ vector/capac\-ities appropriately. We now remove the assumption $\beta\geq 1$ and analyze the running time of Algorithm \textsc{\tt OracleRayApprox\_Route} as a whole. 

\begin{theorem}\label{routing:step2}
For $0< \omega<1$, Algorithm \textsc{\tt OracleRayApprox\_Route} computes a $(1+\omega)A$-approx\-imate solution to the linear program~\eqref{lp:route} in time $O(\omega^{-2} (|\mu|\log A|\mu| + |\epsilon|)A\log |\epsilon| \cdot T_{DSteiner})$, where $T_{DSteiner}$ is the time required to solve the minimum cost directed Steiner tree problem with oracle $\OOO_{DSteiner}$ within an approximation guarantee $A$.
\end{theorem}
\begin{proof}
To remove the dependency of the running time on $\beta$, we update the $\hat{q}$ vector and variable $scaling\_factor$ appropriately. Let $z_i^*$ be the exact maximum fractional Steiner tree packing value for message $m_i$ in line 1 and let $z^* = \min_i \frac{z_i^*}{q_i}$. Then, $z^*$ is an upper bound on the maximum rate at which the messages can be routed in a minimal fractional routing solution. Since $z_i \leq z_i^* \leq A z_i$, $\frac{z}{|\mu|} \leq \beta \leq A z$. We scale the $\hat{q}$ vector as in line 3 so that $1\leq \beta \leq A |\mu|$. The assumption $\beta \geq 1$ is satisfied, but $\beta$ could now be as large as $A |\mu|$. We employ the doubling trick as explained in Section 5.2 in \cite{garg} and line 23 accomplishes this, together with $N = 2\left\lceil\frac{1}{\eta A} \log_{1+\eta} \frac{|\epsilon|}{1-\eta A} \right\rceil$ in the for-loop in line 8. Since we halve the ``current'' value of $\beta$ after every $N$ phases, the total number of phases is at most $N\log A|\mu|$. Since there are $|\mu|$ iterations per phase, there are at most $N|\mu| \log A|\mu|$ iterations in total. Each iteration consists of variably many number of steps and, in all steps but the last, we increase the length of an edge by a factor of $1+\eta$. By similar reasons as in the proof of Claim 5.1 in \cite{garg}, the length of each edge can increase by a factor of $1+\eta$ at most $\log_{1+\eta} \frac{1}{\delta}$ times throughout the algorithm. Hence, the total number of steps of the while-loop in line 11 exceeds the total number of iterations by at most $|\epsilon| \log_{1+\eta} \frac{1}{\delta}$. The total number of steps, hence calls to the oracle $\OOO_{DSteiner}$, is at most $N|\mu| \log A|\mu| + |\epsilon| \log_{1+\eta} \frac{1}{\delta}$. 

Note that $\eta = \Theta(\frac{\omega}{A})$. Then,
\begin{align*}
N|\mu|\log A|\mu| + |\epsilon| \log_{1+\eta} \frac{1}{\delta} &= N|\mu| \log A|\mu| + \frac{|\epsilon|}{\eta A} \log_{1+\eta} \frac{|\epsilon|}{1-\eta A} \\
&= O\left(\frac{|\mu|\log A|\mu|}{\eta A} \log_{1+\eta} \frac{|\epsilon|}{1-\eta A} + \frac{|\epsilon|}{\eta A} \log_{1+\eta} \frac{|\epsilon|}{1-\eta A}\right) \\
& =O\left(\frac{|\mu|\log A|\mu| + |\epsilon|}{\eta A} \cdot \frac{\ln |\epsilon| + \ln(1/(1- \eta A))}{\ln (1+\eta)}\right) \\
&\leq O\left(\frac{|\mu|\log A|\mu| + |\epsilon|}{\eta A} \cdot \frac{\ln |\epsilon| + \ln(1/(1-\eta A))}{\eta - \eta^2/2}\right) \\
&= O\left(\frac{|\mu|\log A|\mu| + |\epsilon|}{\omega} \cdot \frac{\ln |\epsilon| + \ln1/(1-\omega)}{\omega/A - (\omega/A)^2/2}\right) \\
&= O(\omega^{-2} (|\mu| \log A|\mu| + |\epsilon|) A \log |\epsilon|)
\end{align*}
\end{proof}

\subsection{Implementations of Oracle $\OOO_{DSteiner}$}

The oracle $\OOO_{DSteiner}$ solves the minimum cost directed Steiner tree problem (Definition~\ref{prob:DSteiner}). A brute force approach is to loop through all possible subsets of $\epsilon$ and select one with the minimum cost that satisfies the directed Steiner tree conditions. Since checking whether a subset of edges supports a directed path from the source to each receiver node can be done in $O(|\nu|+|\epsilon|)$ time, the brute force approach has the total running time of $O((|\nu|+|\epsilon|) 2^{|\epsilon|})$ and finds an exact optimal solution. We can also compute an $O(k)$-approximate solution by computing a shortest path from the source to each receiver and combining the paths to form a tree, where $k$ is the number of receivers. A shortest path can be computed efficiently and there are many shortest path algorithms. For instance, the Dijkstra algorithm suffices for our purpose and gives an $O(k)$-approximate solution in time $O(|\nu|(|\nu|+ |\epsilon|) \log |\nu|)$ with a binary minheap. There exists an efficient approximation algorithm for the minimum cost directed Steiner tree problem with a significantly better approximation guarantee by Charikar et al.~\cite{charikar}. Charikar et al. designed a family of algorithms that achieves an approximation ratio of $i(i-1)k^{1/i}$ in time $O(n^i k^{2i})$ for any integer $i>1$, where $n$ is the number of nodes and $k$ is the number of receivers. For our problem, $k\leq |\nu|$ and $n\leq |\nu|$ and we get algorithms that achieve an approximation ratio of $i(i-1) |\nu|^{1/i}$ in time $O(|\nu|^{3i})$ for any integer $i>1$. For $i = \log |\nu|$, we obtain an approximation ratio of $O(\log ^2 |\nu|)$ in time of $O(|\nu|^{3 \log |\nu|})$. We summarize with Table~\ref{table:DSteiner}. Note the tradeoff between the approximation ratio $A$ and the running time.

\begin{table}[h!]
\caption{Implementations of oracle $\OOO_{DSteiner}$}
\label{table:DSteiner}
\centering
\begin{tabular}{|l|l|l|}
\hline 
Algorithm for $\OOO_{DSteiner}$ & Approximation Ratio $A$ & Time \\
\hline
BruteForce & 1 & $O((|\nu|+|\epsilon|) 2^{|\epsilon|})$\\
ShortestPathApproximation & $O(|\nu|)$ & $O((|\nu|^2 + |\nu||\epsilon|) \log |\nu|)$\\
Charikar et al.~\cite{charikar} & $i(i-1)|\nu|^{1/i}$ & $O(|\nu|^{3i})$ \\
 & $O(\log^2 |\nu|)$ & $O(|\nu|^{3 \log |\nu|})$\\
\hline
\end{tabular}
\end{table}

\section{Network Linear Coding Capacity Regions}
In this section, we show a computable inner bound on the network linear coding capacity region $\CCC_l$ with respect to a given finite field. In particular, we show how to compute a polytope $\CCC_l'$, which we call the {\em semi-network linear coding capacity region}, that is contained in $\CCC_l$. If $\CCC_l'$ is strictly bigger than $\CCC_r$, then linear coding helps improve the information throughput through the network. It is unknown at this time how good of an approximation the polytope $\CCC_l'$ is to the actual network linear coding capacity region $\CCC_l$. Unlike in the computation of the network routing capacity region, the finite field is important in the computation of $\CCC_l'$. We assume that a network $\NNN$, not necessarily a multiple multicast network as in Section~\ref{sec:routing}, and a finite field $F$ are given in what follows, if not stated explicitly.

\subsection{Definitions}

\begin{definition}[Weight Vectors and Partial Scalar-Linear Network Code Solutions]
Let $\NNN$ be a network with unit edge capacities and $m_1, \ldots, m_{|\mu|}$ be the messages. The {\em weight vectors associated with $\NNN$}, or simply {\em weight vectors}, are vectors $w$ in $\{0,1\}^{|\mu|}$ such that there exists a scalar-linear network code solution for $\NNN$ when only messages $m_i$ with $w_i=1$ are considered, i.e., for $\NNN$ with the new message set $\mu' = \{m_i: w_i = 1\}$. We refer to the scalar-linear network code solutions corresponding to these weight vectors as {\em partial scalar-linear network code solutions}, or {\em partial scalar-linear solutions} for short. 
\end{definition}

Note that by definition, Steiner trees are also partial scalar-linear network code solutions.
 
\begin{definition}[Simple Fractional Linear Network Code Solution]
Let $\NNN=(\nu, \epsilon, \mu, c, \AAA, S,R)$ be a capacitated network and $m_1, \ldots, m_{\lvert \mu \rvert}$ be the messages. A fractional network code $(F, \hat{k}, n, \FFF_e, \FFF_d)$ is a {\em simple fractional linear network code solution}, or {\em simple fractional linear solution} for short, if the fractional network code is linear over the finite field $F$ and can be decomposed into a set of partial scalar-linear solutions of $\NNN$ (when considered with unit edge capacities).
\end{definition}

\begin{definition}[Semi-Network Linear Coding Capacity Region]
Let $\NNN=(\nu, \epsilon, \mu, c, \AAA, S,R)$ be a capacitated network and $m_1, \ldots, m_{\lvert \mu \rvert}$ be the messages. The {\em semi-network linear coding capacity region} $\CCC'_l$ of $\NNN$ is the closure of all coding rate vectors achievable by simple fractional linear network code solutions. Note $\CCC'_l \subset \RR_+^{\lvert \mu \rvert}$. 
\end{definition}

Clearly, the network linear coding capacity region $\CCC_l$ contains the semi-network linear coding capacity region $\CCC'_l$ as the set of fractional linear code solutions is a superset of the set of simple fractional linear code solutions.

\subsection{Properties}

\begin{theorem}
Assume a finite field $F$ is given. The semi-network linear coding capacity region $\CCC'_l$, with respect to $F$, is a bounded rational polytope in $\RR_+^{|\mu|}$ and is computable.
\end{theorem}
\begin{proof}
(Polytope) The proof is similar to the proof of Theorem~\ref{thm:routingregion}. Let $\NNN =(\nu, \epsilon, \mu, c, \AAA,S,R)$ be a network. Let $w_1, \ldots, w_{k'}$ be all the weight vectors associated with $\NNN$. Let $\WWW_i$ be the set of all partial scalar-linear network code solutions that satisfy the demands corresponding to the weight vector $w_i$ and $\WWW$ be the union, $\WWW = \WWW_1 \cup \ldots \cup \WWW_{k'}$. Note that $\WWW_i$'s are finite nonempty disjoint sets. Then, any simple fractional linear code solution $(F, \hat{k}, n, \FFF_e, \FFF_d)$ can be decomposed into partial scalar-linear solutions in $\WWW$ and satisfies the following constraints:
\[
\begin{array}{llll}
\sum_{W\in \WWW} W(e) \cdot x(W) & \leq & c(e)\cdot n, & \forall e\in \epsilon \\
\sum_{i=1}^{k'}\sum_{W\in \WWW_i} [w_i]_j x(W)& = & k_j, & \forall 1\leq j \leq |\mu |\\
 x &\geq & 0, &
\end{array}
\]
where $x(W)$ is the number of times the partial scalar-linear solution $W$ is used in the fractional linear solution and $W(e)$ is an indicator that is 1 if the solution $W$ uses edge $e$, or 0 otherwise. After dividing all the variables $x(W)$ by $n$, it follows that all simple fractional linear code solutions satisfy
\begin{equation}\label{parentpolytopel}
\begin{array}{llll}
\sum_{W\in \WWW} W(e) \cdot x(W) & \leq & c(e), & \forall e\in \epsilon \\
 x &\geq & 0. &
\end{array}
\end{equation}
Using the affine map 
\[
\psi_l:(x(W))_{W\in \WWW} \mapsto \left(\sum_{i=1}^{k'}\sum_{W\in\WWW_i} [w_i]_1 x(W), \ldots, \sum_{i=1}^{k'}\sum_{W\in\WWW_{i}} [w_i]_{|\mu|} x(W)\right),
\]
we follow the similar lines of reasoning as in Theorem~\ref{thm:routingregion} to show that the semi-network linear coding capacity region $\CCC'_l$ is a bounded rational polytope in $\RR_+^{|\mu|}$. 

(Computability) We apply the same proof for $\CCC_r$ to $\CCC_l'$ using the inequalities in \eqref{parentpolytopel}. 
\end{proof}

We use $\PPP'_l$ to denote the ``parent'' polytope of $\CCC'_l$ defined by \eqref{parentpolytopel}.
\subsection{Algorithms}
We obtain algorithms and heuristics for computing the semi-network linear coding capacity region $\CCC'_l$ from algorithms and heuristics in Section~\ref{subsec:routingalgo} with little modifications; we use the polytope description of $\PPP_l$ instead of $\PPP_r$ and use the ray oracle $\OOO_{Ray}$ for $\CCC'_l$. We denote the resulting algorithms by Algorithms \textsc{\tt VertexEnum\_LCode} and \textsc{\tt FacetEnum\_LCode}. We omit the details of the algorithms. This subsection goes together with next two subsections.

\subsection{Implementations of Exact and Approximate Oracle $\OOO_{Ray}$}
\subsubsection{Algorithms}
We provide the implementations of oracle $\OOO_{Ray}$ for the semi-network linear coding capacity region $\CCC'_l$. As the region $\CCC'_l$ is a rational polytope, it suffices to consider rays with a rational slope. Given the hyperplane description of the polytope $\PPP'_l$,
\[
\begin{array}{llll}
\sum_{W\in \WWW} W(e) \cdot x(W) & \leq & c(e), & \forall e\in \epsilon \\
x &\geq & 0, &
\end{array}
\]
and a ray with a rational slope of the form $\hat{x} = \hat{q}t, t\geq 0$, we would like to compute the rational intersection point of the ray and the boundary of polytope $\CCC_l'$. The intersection point is $\lambda_{max}\hat{q}$ where $\lambda_{max}$ is the optimal value to the linear program:
\begin{equation}\label{lp:lcoding}
\begin{array}{lllll}
\max & \lambda & & & \\
\operatorname{s.t.}& \sum_{W\in \WWW} W(e)\cdot x(W) & \leq & c(e), & \forall e\in \epsilon \\
 &\sum_{i=1}^{k'}\sum_{W\in \WWW_i} [w_i]_j x(W) & \geq & \lambda q_j, & \forall 1\leq j \leq |\mu| \\
 &x, \lambda & \geq &0, & \\
\end{array}
\end{equation}
where $w_1, \ldots, w_{k'}$ are the weight vectors associated with network $\NNN$. Since the coefficients of the linear program are rational, the optimal value $\lambda_{max}$ and the corresponding solution $x$ are rational. We can use any linear programming algorithm to solve \eqref{lp:lcoding} exactly, as in Algorithm \textsc{\tt OracleRayExact\_Route}, and obtain Algorithm \textsc{\tt OracleRayExact\_LCode}. We omit the pseudocode.

Using techniques by Garg and K\"{o}nemann~\cite{garg}, we provide a combinatorial approximation algorithm, Algorithm \textsc{\tt OracleRayApprox\_LCode}, for solving the linear program \eqref{lp:lcoding} approximately. While the linear program~\eqref{lp:lcoding} looks similar to the linear program~\eqref{lp:route}, it is much harder to solve. The algorithm computes a point $\hat{r}$ such that $\lambda_{max}\hat{q}$ is between $\hat{r}$ and $(1+\omega)B\hat{r}$, for some numbers $\omega>0$ and $B\geq 1$. We assume we have oracles $\OOO_{SLinear}$ and $\OOO_{FCover}$ for the following two subproblems related to \eqref{lp:lcoding}: 

\begin{definition}[Minimum Cost Scalar-Linear Network Code Problem]\label{prob:SLinear}
Given a network $\NNN =(\nu, \epsilon, \mu, c, \AAA,S, R)$ with unit edge capacities, a finite field $F$ and a length function $l:\epsilon\rightarrow \RR_+$, compute the minimum cost scalar-linear network code solution for $\NNN$ with respect to $F$, if it exists. The cost of a solution is the sum of lengths of the edges used in the solution. If there is no scalar-linear solution, then report ``unsolvable.''
\end{definition}

Without the minimum cost condition, the above problem reduces to the decidability problem of determining whether or not a network has a scalar-linear solution, which is NP-hard by Theorem 3.2 in Lehman and Lehman~\cite{lehman}. Hence, the above problem is at least as hard as any NP-hard problem. We assume that $\OOO_{SLinear}$ solves the minimum cost scalar-linear network code problem exactly.

\begin{definition}[Fractional Covering with Box Constraints Problem]\label{prob:fcover}
Given an $n\times m$ nonnegative integer matrix $A$, a nonnegative vector $b$, a positive vector $c$ and a nonnegative integer vector $u$, compute
\begin{equation*}
\begin{array}{lllll}
\min & \sum_{j=1}^m c(j)x(j) & & & \\
\operatorname{s.t.}& \sum_j A(i,j) x(j) & \geq & b(i) , & \forall 1 \leq i \leq n \\
 & x(j) & \leq & u(j), & \forall 1\leq j \leq m \\
 &x & \geq &0. & 
\end{array}
\end{equation*}
\end{definition}

In Algorithm \textsc{\tt OracleRayApprox\_LCode}, $\OOO_{FCover}$ solves the fractional covering problem of the following form with an approximation guarantee of $B$:
\begin{equation}\label{lp:ourfcover}
\begin{array}{lllll}
\min & \sum_{i=1}^{k'} y(i) U_i(l) & & & \\
\operatorname{s.t.}& \sum_{i=1}^{k'} [w_i]_j y(i)& \geq & q_j, & \forall 1 \leq j \leq |\mu| \\
 &y(j)& \leq &\lceil q_j \rceil,& \forall 1 \leq j \leq k'\\
 &y & \geq &0, & \\
\end{array}
\end{equation}
where $w_1, \ldots, w_{k'}$ are the weight vectors associated with the network and $U_i(l)$ is the cost of the minimum cost partial scalar-linear solution in $\WWW_i$ with respect to the length function $l$.

\begin{algorithm}
\caption{Algorithm \textsc{\tt OracleRayApprox\_LCode}($\NNN$, $\omega$, $\OOO_{SLinear}$, $\OOO_{FCover}$, $B$, $\hat{q}$)}
\begin{algorithmic}[1]
\STATE Using Algorithm \textsc{\tt DSteinerTreePacking}, compute the (approximate) value $z_i$ to \eqref{lp:singlepack} for each message $m_i$ separately.
\STATE $z = \min_i \frac{z_i}{q_i}$; $scaling\_factor = \frac{|\mu|}{z}$
\STATE $\hat{q} = scaling\_factor \cdot \hat{q}$
\STATE $\eta = \frac{1}{9B} \omega$; $\delta = (|\epsilon|/(1-\eta B))^{-1/\eta B}$
\STATE $N = 2\left\lceil\frac{1}{\eta B} \log_{1+\eta} \frac{|\epsilon|}{1- \eta B} \right\rceil$
\WHILE {\TRUE}
	\STATE $t=0$; $l(e) = \frac{\delta}{c(e)}, \forall e\in \epsilon$
	\FOR {phase $i=1,\ldots, N$}
		\STATE $\gamma = 1$
		\WHILE {$\gamma >0$}
			\STATE Using $\OOO_{SLinear}$, compute $U_i(l)$ for each weight vector $w_i$ and corresponding partial scalar-linear solution $\widetilde{W}_i$ with the minimum cost.
			\STATE Using $\OOO_{FCover}$, solve \eqref{lp:ourfcover} to get the $y(1), \ldots, y(k')$ values.
			\STATE $\widetilde{W} = y(1)\widetilde{W}_1 + \cdots + y(k') \widetilde{W}_{k'}$
			\STATE $s = \max \left\{\frac{\widetilde{W}(e)}{c(e)}: e \textrm{ such that } \widetilde{W}(e) > c(e)\right\}$
			\STATE $\widetilde{W} = \frac{1}{s} \widetilde{W}$
			\STATE $\gamma = \gamma - \frac{1}{s} \gamma$
			\STATE Update $l(e) = l(e)\left(1+\eta \frac{\widetilde{W}(e)}{c(e)}\right)$.
			\IF {$\sum_e l(e)c(e) >1$}
				\STATE Goto Line \ref{lastline2}.
			\ENDIF
		\ENDWHILE
		\STATE $t = t+1$
	\ENDFOR
	\STATE $scaling\_factor = 2 scaling\_factor$; $\hat{q} = 2\hat{q}$
\ENDWHILE
\RETURN $scaling\_factor\cdot t / \log_{1+\eta} \frac{1}{\delta}$ \label{lastline2}
\end{algorithmic}
\end{algorithm}

\subsubsection{Analysis}

The corresponding dual linear program of \eqref{lp:lcoding} is 
\begin{equation}\label{lp:lcodingdual}
\begin{array}{lllll}
\min & \sum_e l(e) c(e) & & & \\
\operatorname{s.t.}& \sum_{e\in \epsilon} W(e) l(e) - \sum_{j=1}^{k} [w_i]_j z(j) & \geq & 0, & \forall 1\leq i\leq k', \forall W\in \WWW_i \\
 &\sum_{j=1}^{k} q_j z(j) & \geq & 1 & \\
 &l, z & \geq &0. & \\
\end{array}
\end{equation}

We rewrite the dual linear program~\eqref{lp:lcodingdual} as two recursively nested linear programs. Consider the following linear program derived from the dual program:
\begin{equation}\label{lp:lcoding2}
\begin{array}{lllll}
\max & \sum_{j=1}^k q_j z(j) & & & \\
\operatorname{s.t.}& \sum_{j=1}^k [w_i]_j z(j) & \leq & U_i(l), & \forall 1\leq i\leq k' \\
 &z & \geq &0, & \\
\end{array}
\end{equation}
where $U_i(l) = \min_{W\in \WWW_i} \sum_e W(e) l(e)$. Let $D(l) = \sum_e l(e)c(e)$ and $\alpha(l)$ be the optimal value of the linear program~\eqref{lp:lcoding2}. Then, solving \eqref{lp:lcodingdual} is equivalent to finding an assignment of lengths to the edges, $l:\epsilon \rightarrow \RR^+$, so as to minimize $\frac{D(l)}{\alpha(l)}$. Let $\beta$ denote the optimal value of \eqref{lp:lcodingdual}, i.e., $\beta = \min_l \frac{D(l)}{\alpha(l)}$. The dual linear program for \eqref{lp:lcoding2} is 

\begin{equation}\label{lp:lcoding2dual}
\begin{array}{lllll}
\min & \sum_{i=1}^{k'} y(i) U_i(l) & & & \\
\operatorname{s.t.}& \sum_{i=1}^{k'} [w_i]_j y(i)& \geq & q_j, & \forall 1 \leq j \leq |\mu| \\
 &y & \geq &0. & \\
\end{array}
\end{equation}
Let $\alpha'(y) = \min_{j:q_j\not= 0} \frac{\sum_{i=1}^{k'} [w_i]_j y(i)}{q_j}$ and $D'_l(y) = \sum_{i=1}^{k'} y(i) U_i(l)$. Without loss of generality, we assume that an optimal solution to \eqref{lp:lcoding2dual} satisfies $y(j)\leq \lceil q_j \rceil$ for all $j$; then \eqref{lp:lcoding2dual} is equivalent to \eqref{lp:ourfcover}. Solving \eqref{lp:lcoding2dual} is equivalent to finding an assignment of values to variables $y$, $y:\{1, \ldots, k'\}\rightarrow  \RR^+$, so as to minimize $\frac{D'_l(y)}{\alpha'(y)}$. Let $\beta'$ be the optimal value of \eqref{lp:lcoding2dual}, i.e., $\beta' = \min_y \frac{D'_l(y)}{\alpha'(y)}$. By linear programming duality, $\alpha(l) = \beta'$. Then,
\[
\beta = \min_l \frac{D(l)}{\alpha(l)} =\min_l \frac{D(l)}{\min_y \frac{D'_l(y)}{\alpha'(y)}}.
\]

First, we consider Algorithm \textsc{\tt OracleRayApprox\_LCode} with the infinite while-loop in line 6 removed, with variable $scaling\_factor$ and line 24 removed, and with the finite for-loop in line 8 replaced with an infinite while-loop.
\begin{theorem}\label{lcoding:step1}
Assume $\beta\geq 1$. For $0< \omega <1$, Algorithm \textsc{\tt OracleRayApprox\_LCode}, with modifications as explained above, returns a $(1+\omega)B$-approximate solution to the linear program~\eqref{lp:lcoding} in at most $\left\lceil \frac{\beta}{\eta B} \log_{1+\eta} \frac{|\epsilon|}{1-\eta B}\right\rceil$ number of phases.
\end{theorem}
\begin{proof}
Let $l_{i,j-1}$ denote the length function $l$ and $\gamma_{i,j-1}$ the variable $\gamma$ at the start of the $j$-th iteration (of the while-loop in line 10) in phase $i$. Let $\widetilde{W}_{i,j}$ denote the fractional solution $\widetilde{W}$ computed in the $j$-th iteration in phase $i$. Let $l_{i,0}$ be the length function at the start of phase $i$, or equivalently, at the end of phase $i-1$. For simplicity, we denote $D(l_{i+1, 0})$ and $\alpha(l_{i+1, 0})$ by $D(i)$ and $\alpha(i)$. Note that for each phase $i$ and iteration $j$,
\begin{align*}
D(l_{i,j}) &= \sum_e l(e) c(e) \\
 &= D(l_{i,j-1}) + \eta \sum_e l_{i,j-1}(e) \widetilde{W}_{i,j} (e)\\
 &= D(l_{i,j-1}) + \eta \sum_e l_{i,j-1}(e) (\gamma_{i,j-1}-\gamma_{i,j}) [y(1)\widetilde{W}_1 + \ldots + y(k')\widetilde{W}_{k'}]_e \\
 &\leq D(l_{i,j-1}) + \eta (\gamma_{i,j-1}-\gamma_{i,j}) B \alpha(l_{i,j-1}) \\
 &\leq D(l_{i,j-1}) + \eta (\gamma_{i,j-1}-\gamma_{i,j}) B \alpha(l_{i+1, 0}),
\end{align*}
where $\widetilde{W}_1, \ldots,\widetilde{W}_{k'}$ and $y(1), \ldots, y(k')$ are the partial scalar-linear solutions and variable $y$ from line 11. Note that we used the fact that $\alpha$ is a nondecreasing function to obtain the last inequality. Summing up the above inequality over the iterations, we obtain
\begin{align*}
D(l_{i+1, 0}) &\leq D(l_{i+1, -1}) + \eta (\gamma_{i+1, -1}-\gamma_{i+1, 0}) B \alpha(l_{i+1,0})\\
 &\leq D(l_{i+1, -2}) + \eta (\gamma_{i+1,-2}-\gamma_{i+1,0})B\alpha(l_{i+1,0}) \\
 &\phantom{XXXXXX} \vdots \\
 &\leq D(l_{i,0}) + \eta B \alpha(l_{i+1,0}), 
\end{align*}
where $l_{i+1,-j}$ (similarly, $\gamma_{i+1, -j}$) denote the length function $l$ (variable $\gamma$) at the start of the $j$-th from the last iteration in phase $i$. It follows that $D(i)\leq D(i-1) + \eta B \alpha(i)$. From here, the analysis is the same as Theorem~\ref{routing:step1}, with $B$ replacing $A$. 
\end{proof}

While the running time of Algorithm \textsc{\tt OracleRayApprox\_LCode} depends on $\beta$ by Theorem~\ref{lcoding:step1}, we can reduce/increase it by scaling the $\hat{q}$ vector/capacities appropriately. We now remove the assumption that $\beta \geq 1$ and analyze the running time of Algorithm \textsc{\tt OracleRayApprox\_LCode} as a whole.

\begin{theorem}\label{lcoding:step2}
For $0<\omega< 1$, Algorithm \textsc{\tt OracleRayApprox\_LCode} computes a $(1+\omega)B$-approx\-imate solution to the linear program~\eqref{lp:lcoding} in time $O(\omega^{-2} (\log A|\mu| + |\epsilon|) B \log |\epsilon| \cdot (T_{FCover} + k' T_{SLinear}))$, where $T_{FCover}$ is the time required to solve the fractional covering problem by $\OOO_{FCover}$ within an approximation guarantee $B$ and $T_{SLinear}$ is the time required to solve the minimum cost scalar-linear network code problem exactly by $\OOO_{SLinear}$.
\end{theorem}
\begin{proof}
Let $z_i^*$ be the exact maximum fractional Steiner tree packing value for message $m_i$ in line 1 and let $z^* = \min_i \frac{z_i^*}{q_i}$. Note that $z^*$ is an upper bound on the maximum rate at which the demands can be satisfied by a simple fractional linear code solution. Since $z_i \leq z_i^* \leq Az_i$, $\frac{z}{|\mu|} \leq \beta \leq A z$. We scale the $\hat{q}$ vector and get $1\leq \beta \leq A|\mu|$. The assumption $\beta\geq 1$ is satisfied, but now $\beta$ could be as large as $A|\mu|$. We employ the doubling trick with $N=\lceil\frac{1}{\eta B} \log_{1+\eta} \frac{|\epsilon|}{1+\eta B} \rceil$. Then, the total number of phases is at most $N\log A|\mu|$. Each phase consists of variably many number of iterations of the while-loop in line 10 and, in all iterations except the last, we increase the length of an edge by a factor of $1+\eta$. Hence, the total number of iterations exceeds the total number of phases by at most $|\epsilon| \log_{1+\eta} \frac{1}{\delta}$. The total number of calls to the oracle $\OOO_{FCover}$ is at most $N\log A|\mu| + |\epsilon|\log_{1+\eta} \frac{1}{\delta}$. For each iteration, we call the oracle $\OOO_{SLinear}$ exactly $k'$ times to compute the $U_i(l)$ values and, hence, the total number of calls to the oracle $\OOO_{SLinear}$ is at most $(N\log A|\mu| + |\epsilon|\log_{1+\eta} \frac{1}{\delta}) k'$. From here, the proof follows Theorem~\ref{routing:step2} closely. 
\end{proof}

\subsection{Implementations of Oracles $\OOO_{SLinear}$ and $\OOO_{FCover}$}

In this subsection, we discuss implementations of oracles $\OOO_{SLinear}$ and $\OOO_{FCover}$.

\subsubsection{Oracle $\OOO_{SLinear}$}

We only consider implementations of exact oracle $\OOO_{SLinear}$ for the minimum cost scalar-linear network code problem (Definition~\ref{prob:SLinear}). The assumption that the finite field $F$ is fixed for the computation of semi-network linear coding capacity region is important; it ensures termination of algorithms for $\OOO_{SLinear}$, given that the decidability of the linear coding problem without a fixed finite field is unknown at this time. 

A brute force approach is to loop through all possible subset of active edges and try all possible combinations of global linear coding vectors (see Section~\ref{sec:globallinear}) on these edges in time $O(2^{|\epsilon|} |F|^{|\mu||\epsilon|})$. For each edge $e=(x,y)$, it takes $O(|\mu| I)$ time to check if the global coding vector on $e$ is the span of global coding vectors on in-edges of $x$, where $I$ is the maximum in-degree of any node. The total running time is $O(2^{|\epsilon|} |F|^{|\mu||\epsilon|} |\mu||\epsilon| I)$ and is exponential in not only in $|\epsilon|$, but also in the number of messages, $|\mu|$. The brute force approach requires $O(|F|^{|\mu||\epsilon|})$ space. 

We propose a faster algorithm of our own that solves the minimum cost scalar-linear network code problem using dynamic programming. First, we relabel nodes so that edges go from a lower-numbered node to a higher-numbered one and arrange the nodes in a line in order. This can be done as the network is acyclic and by a topological sort algorithm. Let $n_1, \ldots, n_{|\nu|}$ be the nodes in order. Second, we create states indexed by a triple ($i$, $\epsilon_i$, $\phi_i$), where
\begin{enumerate}
\item $i$ is an integer, $1\leq i \leq |\nu|-1$,
\item $\epsilon_i$ is the set of edges that have the start node in $\{n_1, \ldots, n_i\}$ and the end node in $\{n_{i+1}, \ldots$, $n_{|\nu|} \}$, and 
\item $\phi_i$ is the global coding vectors on edges in $\epsilon_i$.
\end{enumerate}

In each state $(i, \epsilon_i, \phi_i)$, we store the minimum cost scalar-linear solution, if it exists, with the set of global coding vectors $\phi_i$ on $\epsilon_i$ when the network is restricted to nodes $n_1, \ldots, n_i$. In essence, we consider the restricted network of nodes $n_1, \ldots, n_i$ and compute its minimum cost scalar-linear solution based on the minimum cost scalar-linear solutions found on the restricted network of nodes $n_1, 
\ldots, n_{i-1}$. See Algorithm \textsc{\tt OracleSLinear\_DP} for more details. Note that $valid\_next\_state$ and $valid\_prev\_state$ are linked lists.

\begin{theorem}
Algorithm \textsc{\tt OracleSLinear\_DP} solves the minimum cost scalar-linear network code problem in time $O(|F|^{2|\mu| \Phi} |\nu| |\mu| \Phi^2)$ and space $O(|F|^{|\mu| \Phi})$, where $\Phi = \max_i |\epsilon_i|$. 
\end{theorem}
\begin{proof}
Let $\Phi = \max_i |\epsilon_i|$. Then, the number of states with a specific pair of $i$ and $\epsilon_i$ is at most $|F|^{|\mu| \Phi}$ and the time to compute the states of the form $(i,\epsilon_i, \phi_i)$ from the states of the form $(i-1, \epsilon_{i-1}, \phi_{i-1})$ takes $O(|F|^{2|\mu|\Phi}|\mu| \Phi^2)$, by simply looping pairs of the states $(i,\epsilon_i, \phi_i)$ and $(i-1, \epsilon_{i-1}, \phi_{i-1})$ and checking if the global coding vectors in $\phi_i$ follow from those in $\phi_{i-1}$. As $i$ ranges from 1 to $|\nu|-1$, the total running time of the algorithm is $O(|F|^{2|\mu| \Phi} |\nu| |\mu| \Phi^2)$. The total space required is $O(|F|^{|\mu| \Phi})$ if we use the ``sliding window'' trick where we only keep states in two consecutive levels (corresponding to two consecutive values of $i$) at any time. 
\end{proof}

\begin{algorithm}
\caption{Algorithm \textsc{\tt OracleSLinear\_DP}($\NNN$, $l$)}
\begin{algorithmic}[1]
\STATE Sort nodes with a topological sort algorithm: $n_1, \ldots, n_{|\nu|}$.
\STATE $valid\_next\_state \leftarrow \emptyset$; $valid\_prev\_state \leftarrow (0,\emptyset, \emptyset)$, the trivial empty solution.
\FOR {$i=1, \ldots, |\nu|-1$}
	\STATE Update $valid\_prev\_state$ so that only states that satisfy demands on $n_i$ exist in the list.
	\STATE $valid\_next\_state \leftarrow \operatorname{NULL}$
	\FOR {each possible coding vector combination $\phi_{i+1}$ on $\epsilon_{i+1}$}
		\STATE Find, if possible, a state in $valid\_prev\_state$ that leads to a minimum cost solution for $(i+1, \epsilon_{i+1}, \phi_{i+1})$. 
		\STATE If successful, add the state $(i+1, \epsilon_{i+1}, \phi_{i+1})$ to $valid\_next\_state$.
	\ENDFOR
	\STATE $valid\_prev\_start \leftarrow valid\_next\_state$
\ENDFOR
\STATE From $valid\_prev\_state$, pick the minimum cost solution that satisfies all the demands at $n_{|\nu|}$.
\RETURN the minimum cost solution found or ``unsolvable'' if no such solution exists.
\end{algorithmic}
\end{algorithm}

Clearly, the asymptotic behavior of Algorithm \textsc{\tt OracleSLinear\_DP} is much better than that of the brute force approach when $\Phi \ll |\epsilon|$ in the network. We summarize algorithms for oracle $\OOO_{SLinear}$ with Table~\ref{table:slinear}:

\begin{table}[h!]
\centering
\caption{Implementations of oracle $\OOO_{SLinear}$}
\label{table:slinear}
\begin{tabular}{|l|l|l|l|}
\hline 
Algorithm for $\OOO_{SLinear}$ & Approximation Ratio & Time & Space \\
\hline
BruteForce& 1 & $O(2^{|\epsilon|} |F|^{|\mu||\epsilon|} |\mu| |\epsilon| I)$& $O(|F|^{|\mu| |\epsilon|})$\\
Algorithm \textsc{\tt OracleSLinear\_DP} & 1 & $O( |F|^{2|\mu| \Phi} |\nu| |\mu| \Phi^2)$ & $O(|F|^{|\mu| \Phi})$\\
\hline
\end{tabular}
\end{table}

\subsubsection{Oracle $\OOO_{FCover}$}
The oracle $\OOO_{FCover}$ solves the fractional covering problem (Definition~\ref{prob:fcover}). We can solve the problem with a polynomial-time linear program solver such as the ellipsoid algorithm. In our problem of networks, the number of variables of the linear program \eqref{lp:ourfcover} is at most $2^{|\mu|}$ and the number of constraints is $|\mu|$. Hence, a polynomial-time linear program solver will give an exact solution in time polynomial in $2^{|\mu|}$. As the ellipsoid algorithm could be slow in practice and could depend on $2^{|\mu|}$ poorly, the actual running time might not be practical. Fleischer~\cite{fleischer} proposed a combinatorial approximation algorithm for the covering problem that solves within an approximation ratio of $(1+\omega)$ by using $O(\omega^{-2} 2^{|\mu|} \log (c^{\rm T} u))$ calls to an oracle that returns a most violated constraint (where $\omega >0$). Note that Fleischer's algorithm is proposed for nonnegative integer matrix $A$ and nonnegative integer vectors $b$, $c$ and $u$, but the algorithm still works for our formulation of the fractional covering problem. As there are $|\mu|$ constraints and at most $2^{|\mu|}$ variables, the oracle for the most violated constraint can be implemented in time $O(2^{|\mu|} |\mu|)$, and this leads to an algorithm that computes a $(1+\omega)$-approximate solution in time $O(\omega^{-2} 2^{2|\mu|} |\mu| \log (c^{\rm T} u))$. We summarize with Table~\ref{table:fcover}:

\begin{table}[h!]
\centering
\caption{Implementations of oracle $\OOO_{FCover}$}
\label{table:fcover}
\begin{tabular}{|l|l|l|}
\hline 
Algorithm for $\OOO_{cover}$ & Approximation Ratio $B$ & Time \\
\hline
Ellipsoid Algorithm& 1 & $polynomial(2^{|\mu|}, |\mu|)$\\
Fleischer~\cite{fleischer} & $1+\omega$ & $O(\omega^{-2} 2^{2|\mu|} |\mu| \log (c^{\rm T} u))$\\
\hline
\end{tabular}
\end{table}


\section{Examples}
In this section, we provide examples of network capacity regions in the case of two messages. Instead of Algorithm \textsc{\tt FacetEnum\_Route} (or \textsc{\tt FacetEnum\_LCode}), we use Algorithm \textsc{\tt BoundaryTrace2D} as the polytope reconstruction algorithm. Algorithm \textsc{\tt BoundaryTrace2D} is similar to the algorithm given by Cole and Yap~\cite{cole} in that the common main idea is that three collinear boundary points define a face. Algorithm \textsc{\tt BoundaryTrace2D} is different from Cole and Yap~\cite{cole} in that it considers rays that start from the origin. See the pseudocode for the details of Algorithm \textsc{\tt BoundaryTrace2D}. Note $\LLL$ is a linked list. We refer to Appendix~\ref{app:proof1} for a proof of its correctness. The algorithm works both for $\CCC_r$ and $\CCC_l'$ with an appropriate oracle $\OOO_{Ray}$. 

\begin{algorithm}
\caption{Algorithm \textsc{\tt BoundaryTrace2D}($\NNN$, $\OOO_{Ray}$)}
\begin{algorithmic}[1]
\STATE $\LLL\leftarrow \emptyset$
\STATE Using $\OOO_{Ray}$, compute the intersection points on $x=e_i\cdot t, t\geq 0$ for $i=1,2$ and obtain $(x_1, y_1)$ and $(x_2, y_2)$.
\STATE Insert $(x_1, -1)$, $(x_1, y_1)$,$(x_2, y_2)$, and $(-1, y_2)$ onto the head of $\LLL$, in that order.
\STATE $cur\_pointer\leftarrow$ the head of $\LLL$    
\WHILE {there exist three distinct points after $cur\_pointer$}
	\STATE Let $pt1$, $pt2$, $pt3$, and $pt4$ be the four consecutive points starting at $cur\_pointer$.
	\STATE Let line $l_1$ go through $pt1$ and $pt2$, and line $l_2$ go through $pt3$ and $pt4$.
	\IF {intersection point $p$ exists between $l_1$ and $l_2$}
		\STATE Using $\OOO_{Ray}$, compute the boundary point $r$ on $x=pt, t\geq 0$. 
		\IF {$r \not= pt2$ and $r \not= pt3$}
			\STATE Insert point $r$ into $\LLL$ between $pt2$ and $pt3$.
		\ELSE
			\STATE Advance $cur\_pointer$.
		\ENDIF
	\ELSE
		\STATE Advance $cur\_pointer$.
	\ENDIF
\ENDWHILE
\RETURN	$\LLL \cup (0,0)$, except $(x_1,-1)$ and $(-1,y_2)$.
\end{algorithmic}
\end{algorithm}

The example networks are given in Figure~\ref{fig:samplenetworks}. Nodes 1 and 2 are the source nodes and nodes 5 and 6 are the receiver nodes. To compute the exact description of $\CCC_r$ or $\CCC_l'$, we hard-coded the corresponding linear programs and used a linear program solver, linprog, in MATLAB. As the networks are simple, it was easy to enumerate all Steiner trees and partial scalar-linear solutions, and the corresponding linear programs were small. While this approach worked for the particular examples we present, it might not be suitable for bigger or more complicated networks. To get around the numerical issues, we used the tolerance of .05; for instance, two points whose corresponding coordinates differ by at most .05 are considered the same point. For the networks, we assume $\AAA=\{0,1\}$. For the semi-network linear coding capacity regions, we assume the finite field $F$ is $\FF_2$. To obtain approximate intersections points, we simply used an approximate oracle $\OOO_{Ray}$ in place of the linear program solver. We note that the approximate oracle $\OOO_{Ray}$ worked well with Algorithm \textsc{\tt BoundaryTrace2D} and led to its successful termination for these networks, but this may not hold for arbitrary networks in general.

\begin{figure}
	\centering
	\begin{minipage}{2.0in}
	\centering
	\includegraphics[width=1.5in]{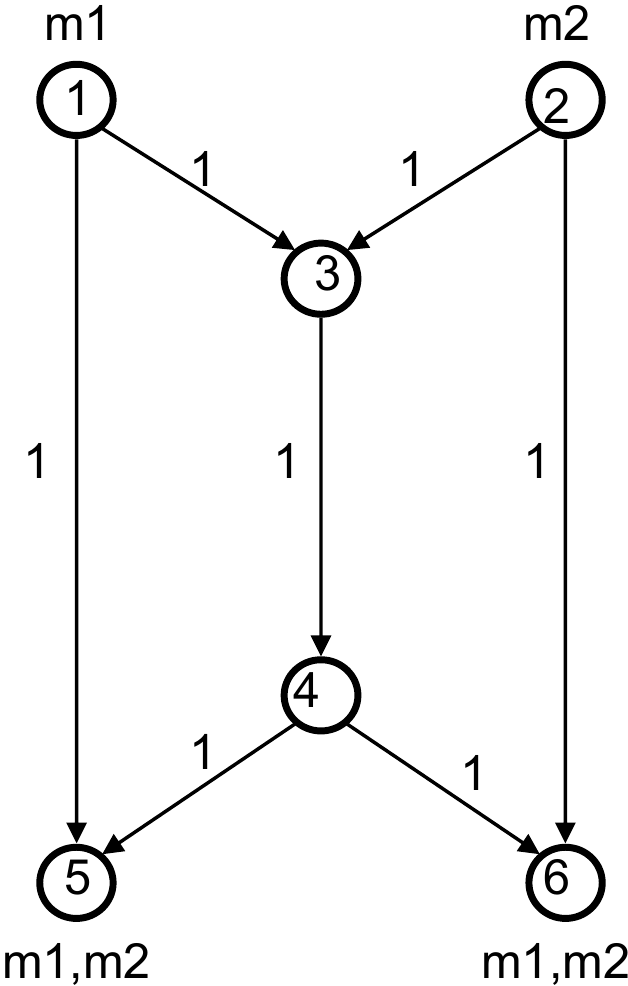}
	\label{sample1:diag}
	\end{minipage}
	\quad
	\begin{minipage}{2.0in}
	\centering
	\includegraphics[width=1.5in]{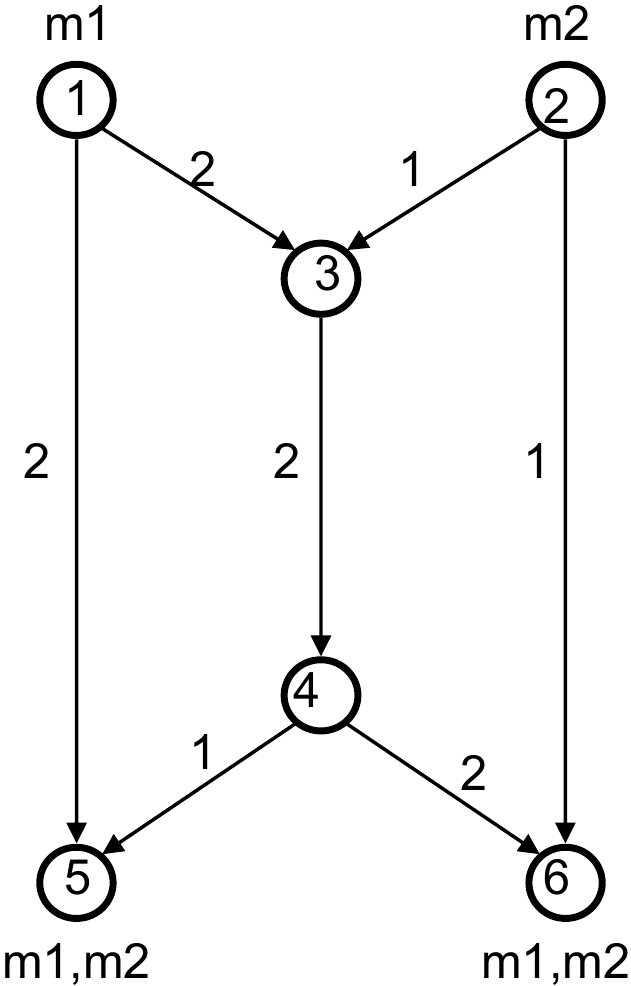}
	\label{sample2:diag}
	\end{minipage}
	\quad
	\begin{minipage}{2.0in}
	\centering
	\includegraphics[width=1.5in]{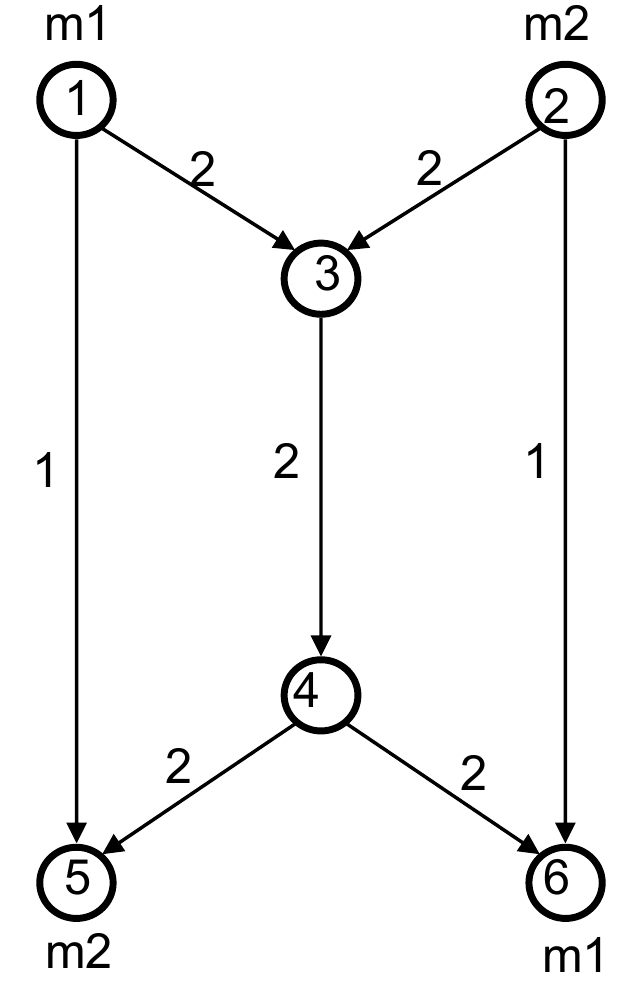}
	\label{sample3:diag}
	\end{minipage}
	\caption{The Butterfly network with unit edge capacities and its variants: $\NNN_1$, $\NNN_2$, and $\NNN_3$.}
	\label{fig:samplenetworks}
\end{figure}

\begin{figure}[!t]
\parbox[!t]{0.45\textwidth }{
\includegraphics[width=2.9in]{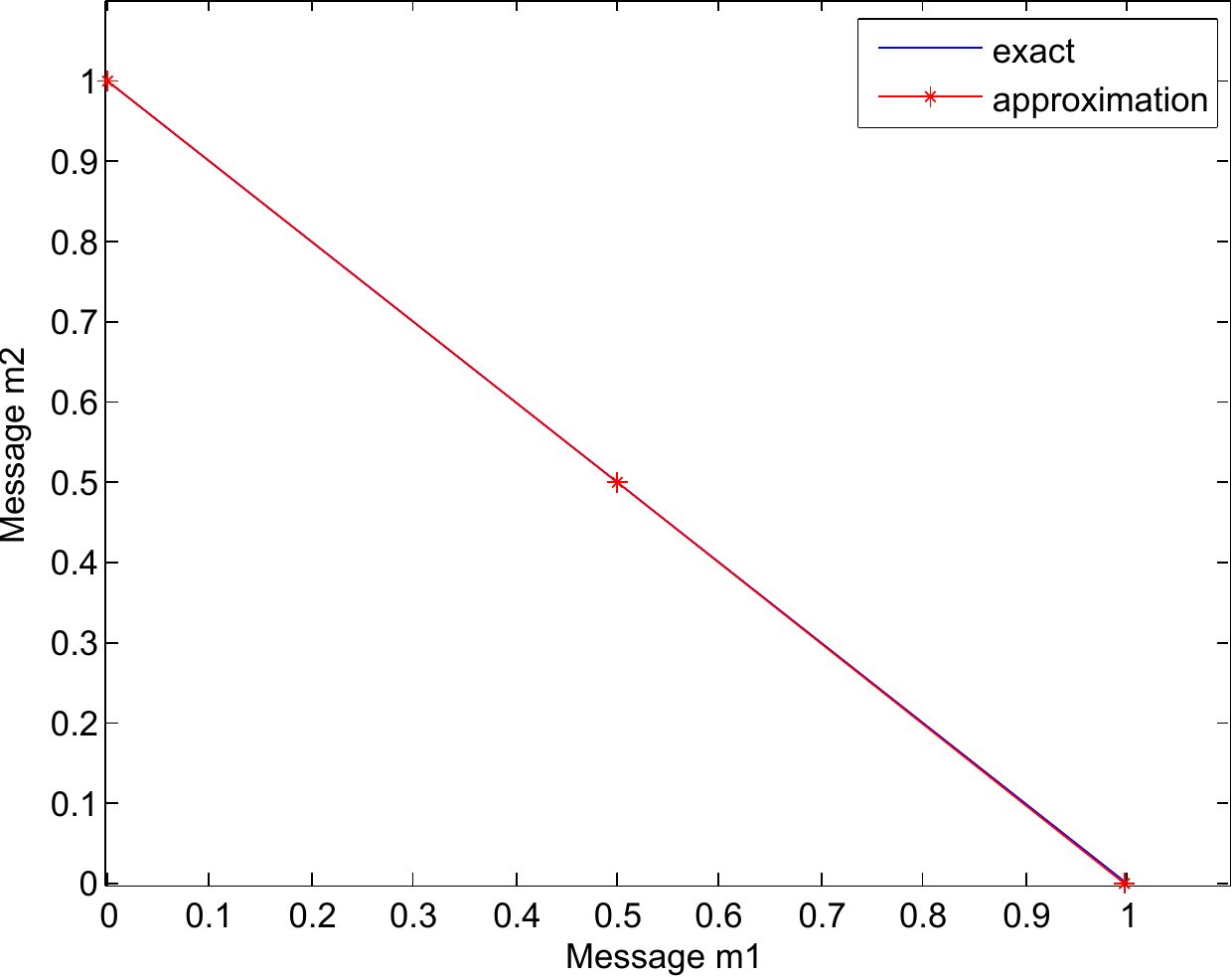}
\captionof{figure}{Network routing capacity region $\CCC_r$ of $\NNN_1$}
\label{sample1:route}
}
\hfill
\parbox[!t]{0.45\textwidth }{
\includegraphics[width=2.9in]{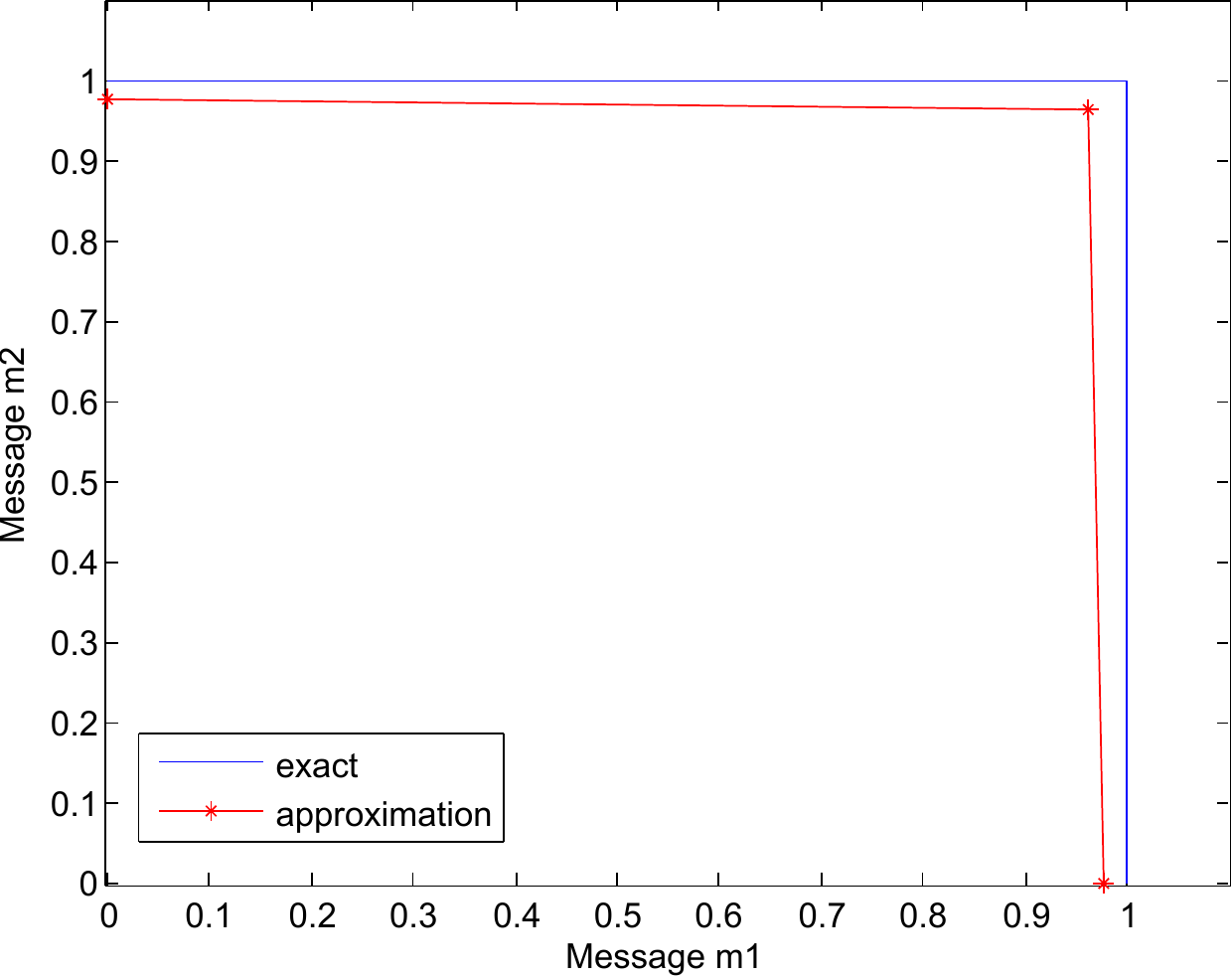}
\captionof{figure}{Semi-network linear coding capacity region $\CCC_l'$ of $\NNN_1$, with respect to $\FF_2$}
\label{sample1:lcoding}
}
\end{figure}

For the network routing capacity region of network $\NNN_1$ in Figure~\ref{sample1:route}, we used $\omega=.5$ in Algorithms \textsc{\tt OracleRayApprox\_Route} and \textsc{\tt DSteinerTreePacking} for $\OOO_{Ray}$ and the algorithm due to Charikar et al.~\cite{charikar} for $\OOO_{DSteiner}$ with the approximation ratio $A=O(\log^2 |\nu|)$. For the semi-network linear coding capacity region of network $\NNN_1$ in Figure~\ref{sample1:lcoding}, we used $\omega=.5$ in Algorithms \textsc{\tt OracleRayApprox\_Route} and \textsc{\tt DSteinerTreePacking} for $\OOO_{Ray}$, the algorithm due to Fleischer~\cite{fleischer} with the approximation ratio $B=1.1$ for $\OOO_{FCover}$, and Algorithm \textsc{\tt OracleSLinear\_DP} for $\OOO_{SLinear}$.

\begin{figure}[!t]
\parbox[!t]{0.45\textwidth }{
\includegraphics[width=2.9in]{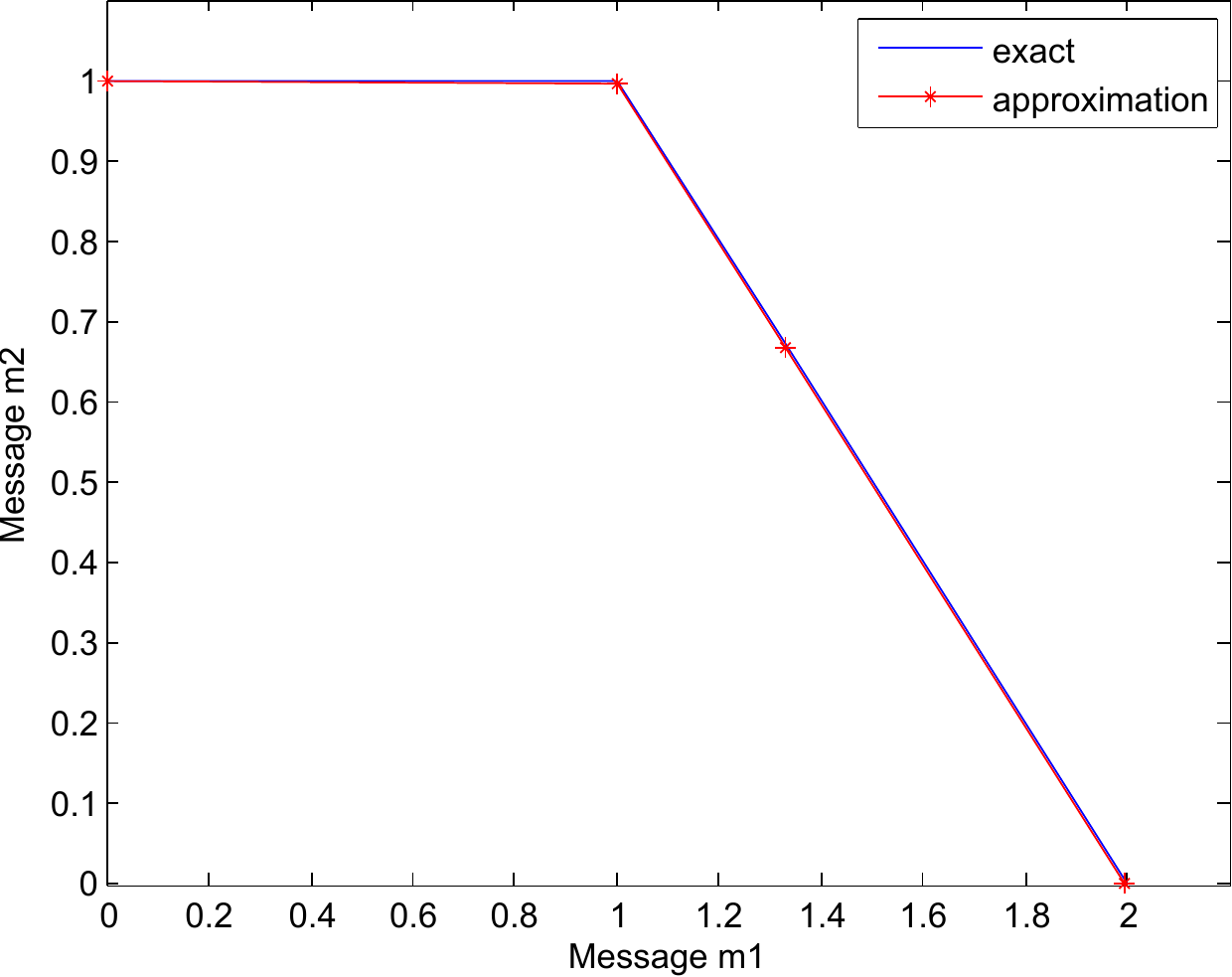}
\captionof{figure}{Network routing capacity region $\CCC_r$ of $\NNN_2$}
\label{sample2:route}
}
\hfill
\parbox[!t]{0.45\textwidth }{
\includegraphics[width=2.9in]{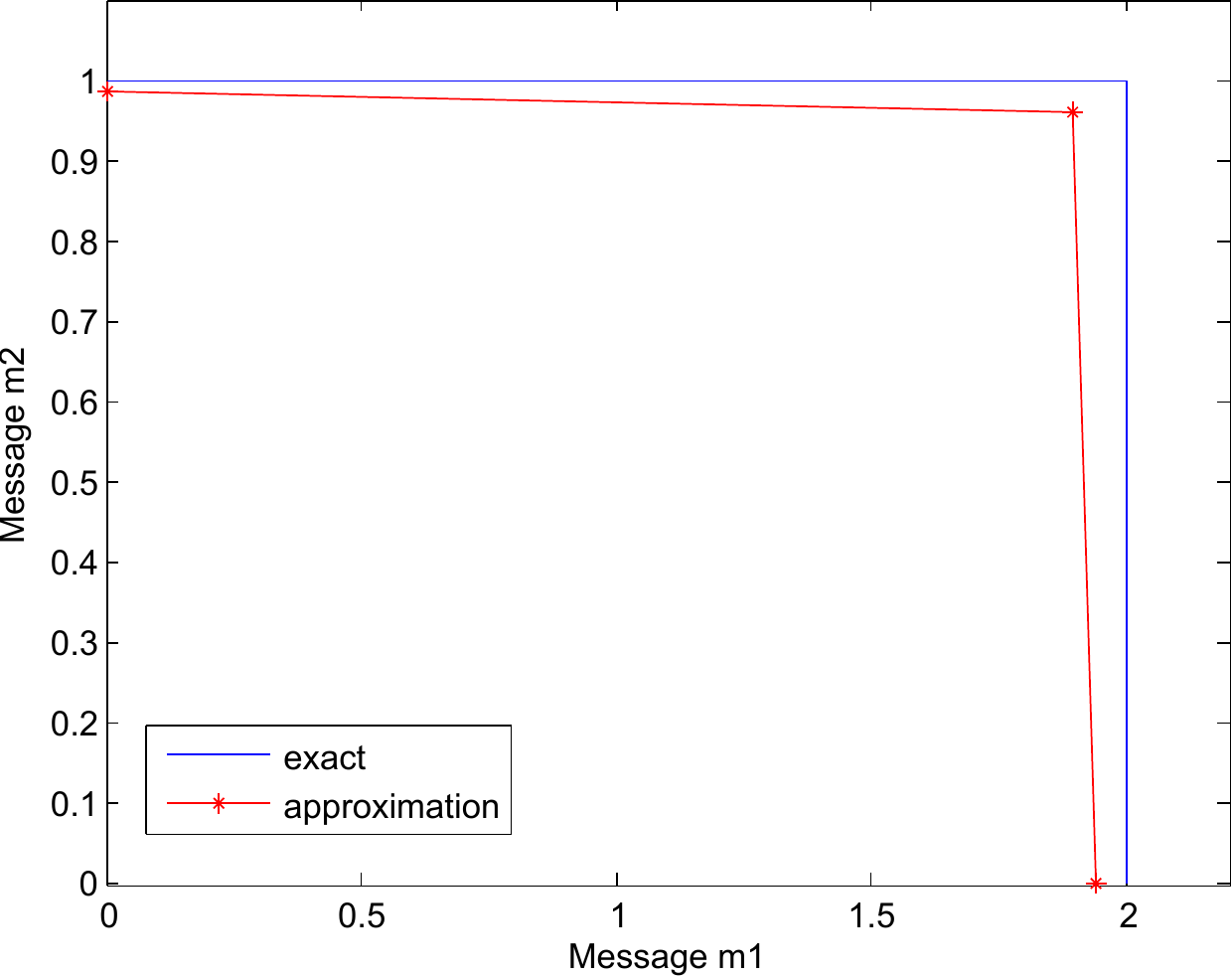}
\captionof{figure}{Semi-network linear coding capacity region $\CCC_l'$ of $\NNN_2$, with respect to $\FF_2$}
\label{sample2:lcoding}
}
\end{figure}

For the network routing capacity region of network $\NNN_2$ in Figure~\ref{sample2:route}, we used $\omega=.5$ in Algorithms \textsc{\tt OracleRayApprox\_Route} and \textsc{\tt DSteinerTreePacking} for $\OOO_{Ray}$ and the algorithm due to Charikar et al.~\cite{charikar} for $\OOO_{DSteiner}$ with the approximation ratio $A=O(\log^2 |\nu|)$. For the semi-network linear coding capacity region of network $\NNN_2$ in Figure~\ref{sample2:lcoding}, we used $\omega=.5$ in Algorithms \textsc{\tt OracleRayApprox\_Route} and \textsc{\tt DSteinerTreePacking} for $\OOO_{Ray}$, the algorithm due to Fleischer~\cite{fleischer} with the approximation ratio $B=1.1$ for $\OOO_{FCover}$, and Algorithm \textsc{\tt OracleSLinear\_DP} for $\OOO_{SLinear}$.

\begin{figure}[!t]
\parbox[!t]{0.45\textwidth }{
\includegraphics[width=2.9in]{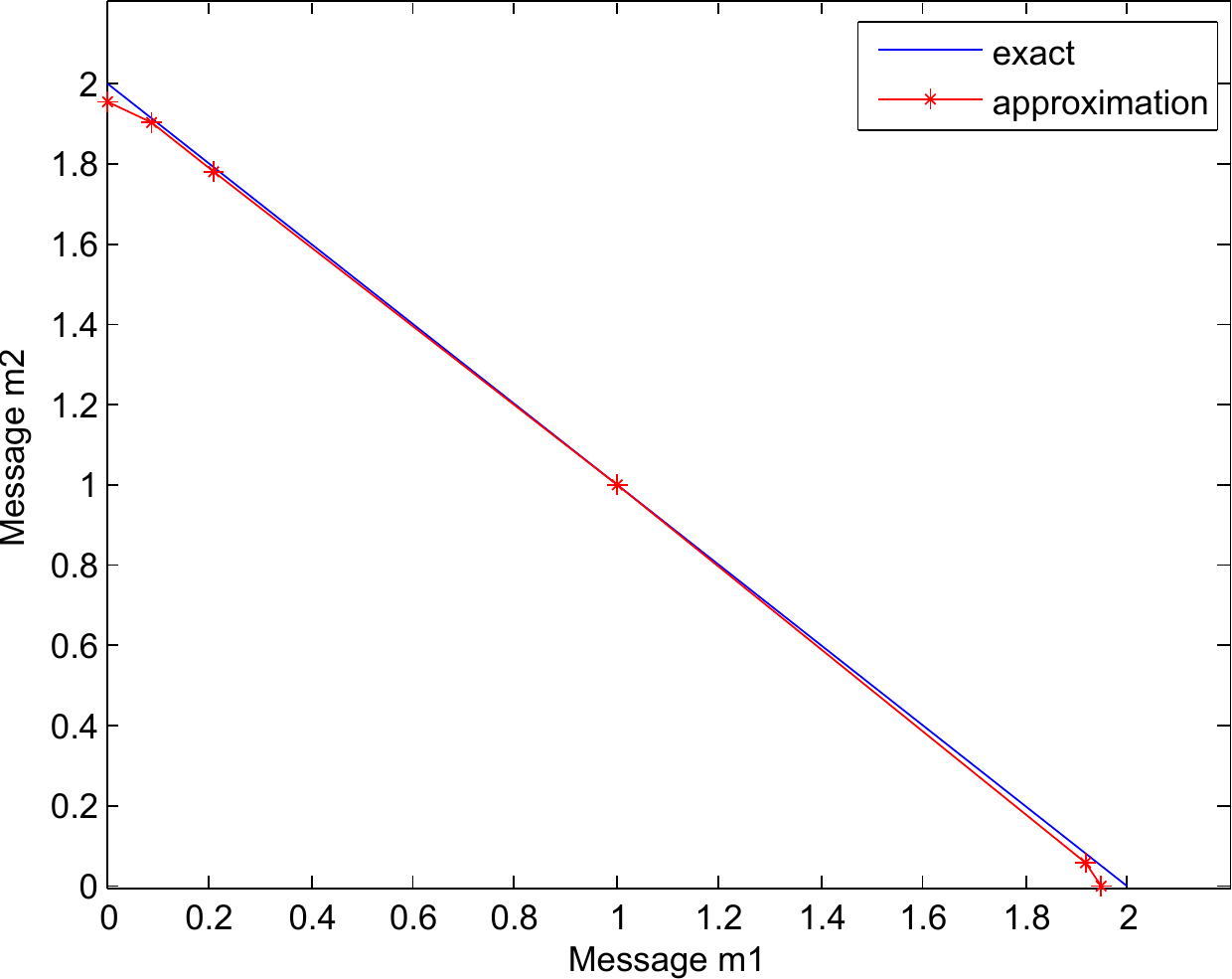}
\captionof{figure}{Network routing capacity region $\CCC_r$ of $\NNN_3$}
\label{sample3:route}
}
\hfill
\parbox[!t]{0.45\textwidth }{
\includegraphics[width=2.9in]{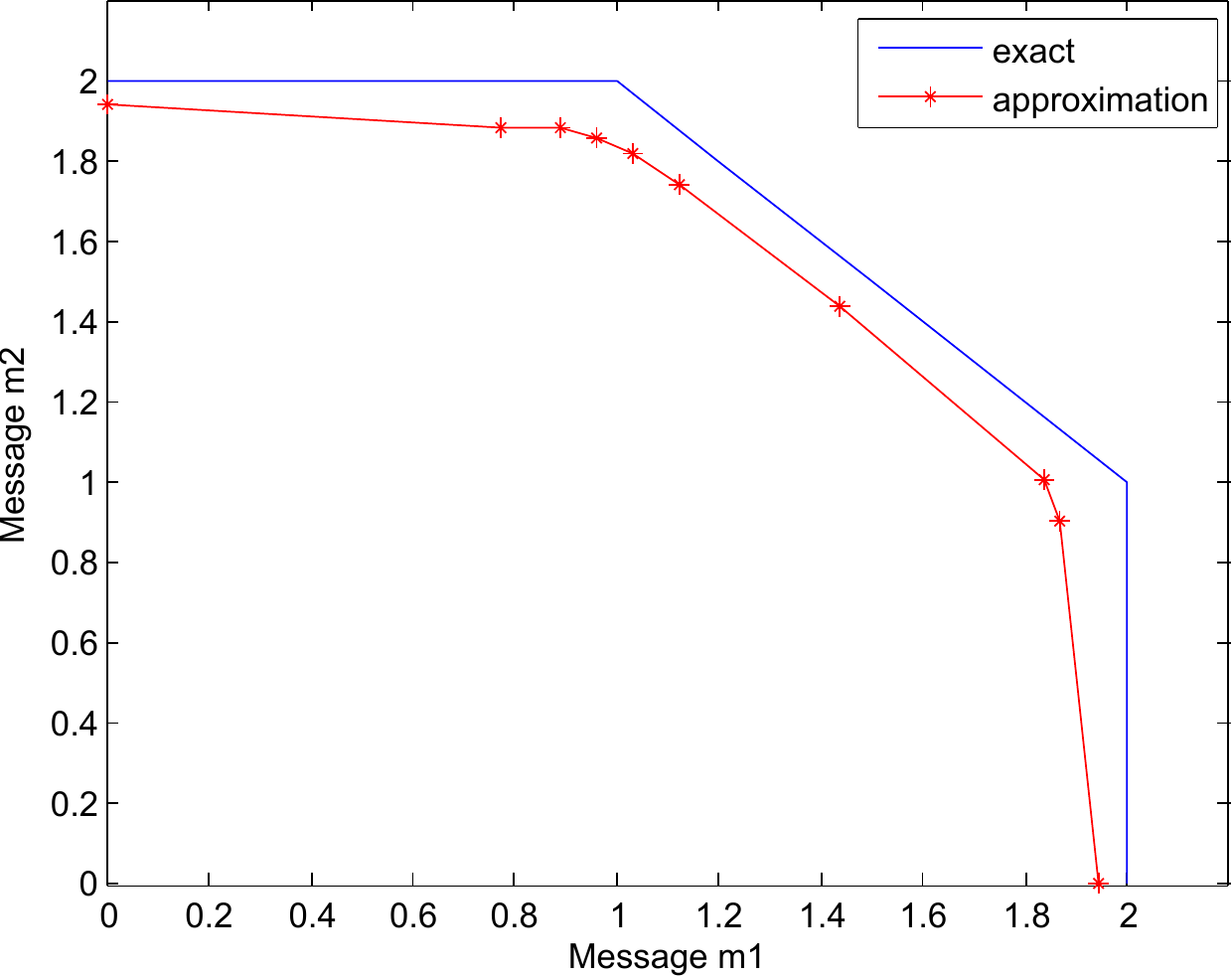}
\captionof{figure}{Semi-network linear coding capacity region $\CCC_l'$ of $\NNN_3$, with respect to $\FF_2$}
\label{sample3:lcoding}
}
\end{figure}

For the network routing capacity region of network $\NNN_3$ in Figure~\ref{sample3:route}, we used $\omega=.5$ in Algorithms \textsc{\tt OracleRayApprox\_Route} and \textsc{\tt DSteinerTreePacking} for $\OOO_{Ray}$ and the brute force algorithm for $\OOO_{DSteiner}$ with the approximation ratio $A=1$. For the semi-network linear coding capacity region of network $\NNN_3$ in Figure~\ref{sample3:lcoding}, we used $\omega=.9$ in Algorithm \textsc{\tt OracleRayApprox\_Route}, $\omega=.5$ in Algorithm \textsc{\tt DSteinerTreePacking}, the algorithm due to Fleischer~\cite{fleischer} with the approximation ratio $B=1.1$ for $\OOO_{FCover}$, and Algorithm \textsc{\tt OracleSLinear\_DP} for $\OOO_{SLinear}$. 

\section{Discussion and Conclusion}
In this chapter, we defined the network capacity region of networks analogously to the rate regions in information theory. In the case of the network routing capacity region, we showed that the region is a rational polytope and provided exact algorithms and approximation heuristics for computing the polytope. In the case of the network linear coding capacity region, we defined an auxiliary polytope, the semi-network linear coding capacity region, that is a rational polytope and that inner bounds the network linear coding capacity region. We provided exact algorithms and approximation heuristics for computing the auxiliary polytope. We noted that the algorithms and heuristics presented in this chapter are not polynomial time schemes. 

Our results have a few straightforward extensions. We can design membership algorithms that given a rate vector, determines whether or not there exists a fractional network code solution that achieves it from algorithms we provided, for the network routing capacity region and semi-network linear coding capacity region. Also, we can compute corresponding approximate solutions $(x(T))_{T\in \TTT}$ for linear program \eqref{lp:route} by storing counters for Steiner trees $\widetilde{T}$ in Algorithm \textsc{\tt OracleRayApprox\_Route}; the same is true for $(x(W))_{W\in \WWW}$.

While our results apply to the networks defined on directed acyclic multigraphs, they generalize to directed networks with cycles and undirected networks straightforwardly. In the computation of the network routing capacity region, we consider minimal fractional routing solutions that can be decomposed into a set of Steiner trees defined appropriately for directed networks with cycles (or undirected networks) and modify the algorithms correspondingly. In the computation of the semi-network linear coding capacity region, we consider simple fractional linear coding solutions that can be decomposed into a set of partial scalar-linear solutions defined appropriately for the networks and modify the algorithms correspondingly.

In connection to Cannons et al.~\cite{cannons:routing}, our work essentially addresses a few problems proposed by Cannons et al.: whether there exists an efficient algorithm for computing the network routing capacity and whether there exists an algorithm for computing the network linear coding capacity. It follows from our work that there exist combinatorial approximation algorithms for computing the network routing capacity and for computing a lower bound of the network linear coding capacity.

We conclude with a few open problems related to our work: determine how good of an inner bound the semi-network linear coding capacity region is to the network linear coding capacity region; design an efficient algorithm, if possible, for computing the linear coding capacity region and network capacity region; design an efficient algorithm, if possible, for the minimum cost scalar-linear network code problem.

\appendix
\chapter{Computations and Proofs}
\section{Computations}
\subsection{Computation in Theorem~\ref{thm:singlepack}}\label{app:comp1}
From Section 2 of Garg and K\"{o}nemann~\cite{garg} and the fact that oracle $\OOO_{DSteiner}$ is an approximate oracle with the approximation guarantee of $A$, it follows that the ratio of dual optimal and primal feasible solutions, $\zeta$, satisfies 
\[
\zeta < (1-\eta)^{-2} A.
\]
We choose $\eta$ appropriately to make sure that the $\zeta< (1+\omega)A$. It suffices to choose $\eta$ such that $(1-\eta)^{-2} \leq 1+\omega$, or equivalently, $(1-\eta)^{-1} \leq (1+\omega)^{1/2}$. By the Taylor Series Theorem, for $0< \omega <1$, we have 
\[
1+\frac{1}{2}\omega - \frac{1}{8} \omega^2 \leq (1+\omega)^{1/2}.
\]
Note that for $0< \eta \leq \frac{1}{2}$, 
\[
(1-\eta)^{-1} = 1 + \eta + \eta^2 + \ldots = 1 + \eta \frac{1}{1-\eta} \leq 1+ 2\eta.
\]
Then, for $\eta = \frac{3}{16} \omega$, we have that $0 < \omega<1$ implies $0< \eta\leq \frac{1}{2}$ and that
\[
(1-\eta)^{-1} \leq 1+ 2\eta \leq 1+\frac{3}{8}\omega \leq 1+\frac{1}{2}\omega - \frac{1}{8} \omega^2 \leq (1+\omega)^{1/2}.
\]

\subsection{Computation in Theorem~\ref{routing:step2}}\label{app:comp2}
We want to choose $\eta$ appropriately so that $(1-\eta A)^{-3} \leq  (1+\omega)$, or equivalently, that $(1-\eta A)^{-1} \leq  (1+\omega)^{1/3}$. By the Taylor Series Theorem, for $0<\omega <1$, we have 
\[
1+\frac{1}{3}\omega - \frac{1}{9} \omega^2 \leq (1+\omega)^{1/3}.
\]
Note that for $0< \eta A \leq \frac{1}{2}$, 
\[
(1-\eta A)^{-1} = 1 + \eta A + (\eta A)^2 + \ldots = 1 + \eta A \frac{1}{1-\eta A} \leq 1+ 2\eta A.
\]
Then, for $\eta = \frac{1}{9 A} \omega$, we have that $0< \omega < 1$ implies $0< \eta A\leq \frac{1}{2}$ and that
\[
(1-\eta A)^{-1} \leq 1+ 2\eta A \leq 1+\frac{2}{9}\omega \leq 1+\frac{1}{3}\omega - \frac{1}{9} \omega^2 \leq (1+\omega)^{1/3}.
\]

\section{Proofs}
\subsection{Proof for Algorithm \textsc{\tt BoundaryTrace2D}}\label{app:proof1}

Without loss of generality, we prove that Algorithm \textsc{\tt BoundaryTrace2D} is correct for the computation of the network routing capacity region $\CCC_r$. Note that $\LLL$ is a linked list of computed boundary points (and the two auxiliary points $(x_1, -1)$ and $(-1, y_2)$) ordered clockwise. When two lines intersect, we mean that the lines intersect in exactly one point. When a line goes through a line segment, we mean that the line intersects the line segment in exactly one point. 

\begin{theorem}\label{thm:bdtrace}
Algorithm \textsc{\tt BoundaryTrace2D} calls the oracle $\OOO_{Ray}$ $O(n)$ times where $n$ is the number of edges in the polygon $\CCC_r$.
\end{theorem}

\begin{proof}
We first show that, for each edge $e$ of $\CCC_r$, the number of distinct boundary points computed in the interior of the edge (excluding the vertices) is at most 3 throughout the execution of Algorithm \textsc{\tt BoundaryTrace2D}. Let edge $e$ have 3 distinct boundary points computed in its interior: $p_1$, $p_2$, and $p_3$, ordered clockwise. Then, any 4 consecutive boundary points in $\LLL$ cannot produce a ray in line 9 that goes through the interior of edge $e$. If the set of 4 consecutive points contains at most 2 of $p_1$, $p_2$, and $p_3$, then, clearly, the ray does not go through the interior of $e$. If the set of 4 consecutive points contains all 3 points $p_1, p_2$, and $p_3$ and a point before $p_1$ in $\LLL$, then the intersection point $p$ in line 8 is exactly $pt2$ and no new boundary point is created. Similarly, it can be shown in other remaining cases that no new boundary point is introduced.

Note that, for each boundary point $b$ in $\LLL$, there can be at most 3 calls to oracle $\OOO_{Ray}$ associated with it; a call to compute the boundary point for the first time, a call if $b$ appears as $pt2$ in line 6 and as the boundary point $r$ in line 9, and a call if $b$ appears as $pt3$ in line 6 and as the boundary point $r$ in line 9. As there are $n-2$ edges and $n-1$ vertices on the outer boundary of $\CCC_r$ and $O(1)$ calls to $\OOO_{Ray}$ for each boundary point computed, the statement follows.
\end{proof}

The correctness of \textsc{\tt BoundaryTrace2D} follows from the fact that each vertex on the outer boundary of $\CCC_r$ is computed by oracle $\OOO_{Ray}$ and that after enough boundary points have been computed, the $cur\_pointer$ in the algorithm will advance to termination. It is easy to see that all vertices of the polygon $\CCC_r$ are included in the returned list $\LLL$. Assume a vertex $v$ is missed in $\LLL$ and the algorithm terminated successfully. Assume that $pt1$, $pt2$, $pt3$ and $pt4$ are the four consecutive points in the resulting list such that the line segment between the origin and $v$ and the line segment between $pt2$ and $pt3$ intersect. Note that $pt2$ and $pt3$ are on different edges of $\CCC_r$. Then, it is easy to see that we have the ray in line 9 going through the interior of the line segment between $pt2$ and $pt3$ and, hence, a new boundary point would have been added. Therefore, $cur\_pointer$ should not have advanced to termination, and this contradicts that the vertex $v$ is missing from $\LLL$. 

\clearpage
\newpage

\begin{singlespace}
\bibliography{main}

\newcommand{\noopsort}[1]{} \newcommand{\printfirst}[2]{#1}
  \newcommand{\singleletter}[1]{#1} \newcommand{\switchargs}[2]{#2#1}
\begin{thebibliography}{10}

\bibitem{ahlswede:network}
R.~Ahlswede, Ning Cai, S.-Y.R. Li, and R.W. Yeung.
\newblock Network information flow.
\newblock {\em Information Theory, IEEE Transactions on}, 46(4):1204--1216, Jul
  2000.

\bibitem{cannons:routing}
J.~Cannons, R.~Dougherty, C.~Freiling, and K.~Zeger.
\newblock Network routing capacity.
\newblock In {\em Information Theory, 2005. ISIT 2005. Proceedings.
  International Symposium on}, pages 11 --13, 4-9 2005.

\bibitem{chan:capacity}
T.~Chan and A.~Grant.
\newblock Mission impossible: Computing the network coding capacity region.
\newblock In {\em Information Theory, 2008. ISIT 2008. IEEE International
  Symposium on}, pages 320 --324, 6-11 2008.

\bibitem{charikar}
Moses Charikar, Chandra Chekuri, To-yat Cheung, Zuo Dai, Ashish Goel, Sudipto
  Guha, and Ming Li.
\newblock Approximation algorithms for directed steiner problems.
\newblock In {\em SODA '98: Proceedings of the ninth annual ACM-SIAM symposium
  on Discrete algorithms}, pages 192--200, Philadelphia, PA, USA, 1998. Society
  for Industrial and Applied Mathematics.

\bibitem{cole}
Richard Cole and Chee~K. Yap.
\newblock Shape from probing.
\newblock {\em J. Algorithms}, 8(1):19--38, 1987.

\bibitem{dougherty:insuff}
R.~Dougherty, C.~Freiling, and K.~Zeger.
\newblock Insufficiency of linear coding in network information flow.
\newblock {\em Information Theory, IEEE Transactions on}, 51(8):2745--2759,
  Aug. 2005.

\bibitem{dougherty:construction}
R.~Dougherty, C.~Freiling, and K.~Zeger.
\newblock Matroidal networks.
\newblock In {\em Allerton Conference on Communication, Control, and
  Computing}, September 2007.

\bibitem{dougherty:matroid}
R.~Dougherty, C.~Freiling, and K.~Zeger.
\newblock Networks, matroids, and non-shannon information inequalities.
\newblock {\em Information Theory, IEEE Transactions on}, 53(6):1949--1969,
  June 2007.

\bibitem{dougherty:poly}
R.~Dougherty, C.~Freiling, and K.~Zeger.
\newblock Linear network codes and systems of polynomial equations.
\newblock {\em Information Theory, IEEE Transactions on}, 54(5):2303--2316, May
  2008.

\bibitem{rouayheb:matroid}
S.~El~Rouayheb, A.~Sprintson, and C.~Georghiades.
\newblock A new construction method for networks from matroids.
\newblock In {\em Information Theory, 2009. ISIT 2009. IEEE International
  Symposium on}, pages 2872--2876, 28 2009-July 3 2009.

\bibitem{fleischer}
Lisa Fleischer.
\newblock A fast approximation scheme for fractional covering problems with
  variable upper bounds.
\newblock In {\em SODA '04: Proceedings of the fifteenth annual ACM-SIAM
  symposium on Discrete algorithms}, pages 1001--1010, Philadelphia, PA, USA,
  2004. Society for Industrial and Applied Mathematics.

\bibitem{fukuda}
Komei Fukuda.
\newblock Frequently asked questions in polyhedral computation.
\newblock \url{http://www.ifor.math.ethz.ch/~fukuda/polyfaq/polyfaq.html}, June
  2004.
\newblock Date retrieved: July 14, 2010.

\bibitem{garg}
Naveen Garg and Jochen K\"{o}nemann.
\newblock Faster and simpler algorithms for multicommodity flow and other
  fractional packing problems.
\newblock {\em SIAM Journal on Computing}, 37(2):630--652, 2007.

\bibitem{handbook}
Jacob~E. Goodman and Joseph O'Rourke, editors.
\newblock {\em Handbook of discrete and computational geometry}.
\newblock CRC Press, Inc., 2 edition, 2004.

\bibitem{gritzmann}
Peter Gritzmann, Victor Klee, and John Westwater.
\newblock {Polytope Containment and Determination by Linear Probes}.
\newblock {\em Proc. London Math. Soc.}, s3-70(3):691--720, 1995.

\bibitem{harvey:capacity}
Nicholas J.~A. Harvey, Robert Kleinberg, and April~Rasala Lehman.
\newblock On the capacity of information networks.
\newblock {\em IEEE/ACM Trans. Netw.}, 14(SI):2345--2364, 2006.

\bibitem{jain}
Kamal Jain, Mohammad Mahdian, and Mohammad~R. Salavatipour.
\newblock Packing steiner trees.
\newblock In {\em SODA '03: Proceedings of the fourteenth annual ACM-SIAM
  symposium on Discrete algorithms}, pages 266--274, Philadelphia, PA, USA,
  2003. Society for Industrial and Applied Mathematics.

\bibitem{kaibel}
Volker Kaibel and Marc~E. Pfetsch.
\newblock Some algorithmic problems in polytope theory.
\newblock In {\em IN ALGEBRA, GEOMETRY, AND SOFTWARE SYSTEMS}, pages 23--47.
  Springer-Verlag, 2003.

\bibitem{kakhbod:routing}
Ali Kakhbod and S.~M. Sadegh~Tabatabaei Yazdi.
\newblock On describing the routing capacity regions of networks.
\newblock {\em Mathematical Methods of Operations Research}, pages 1432--2994,
  May 2010.

\bibitem{koetter:algebra}
R.~Koetter and M.~M\'{e}dard.
\newblock An algebraic approach to network coding.
\newblock {\em Networking, IEEE/ACM Transactions on}, 11(5):782--795, Oct.
  2003.

\bibitem{lehman}
April~Rasala Lehman and Eric Lehman.
\newblock Complexity classification of network information flow problems.
\newblock In {\em Proceedings of the ACM-SIAM Symposium on Discrete
  Algorithms}, pages 142--150, Jan. 2004.

\bibitem{li:linear}
S.-Y.R. Li, R.W. Yeung, and Ning Cai.
\newblock Linear network coding.
\newblock {\em Information Theory, IEEE Transactions on}, 49(2):371--381, Feb.
  2003.

\bibitem{medard:mnetwork}
M.~M\'{e}dard, M.~Effros, T.~Ho, and D.~Karger.
\newblock On coding for non-multicast networks.
\newblock In {\em Proc. 41st Annual Allerton Conference on Communication,
  Control and Computing}, Oct. 2003.

\bibitem{oxley:matroid}
J.G. Oxley.
\newblock {\em Matroid Theory}.
\newblock New York: Oxford Univ. Press, 1992.

\bibitem{sun:networkmatroid}
Qifu Sun, Siu~Ting Ho, and S.-Y.R. Li.
\newblock On network matroids and linear network codes.
\newblock In {\em Information Theory, 2008. ISIT 2008. IEEE International
  Symposium on}, pages 1833 --1837, 6-11 2008.

\bibitem{thakor:capacity}
S.~Thakor, A.~Grant, and T.~Chan.
\newblock Network coding capacity: A functional dependence bound.
\newblock In {\em Information Theory, 2009. ISIT 2009. IEEE International
  Symposium on}, pages 263 --267, june 2009.

\bibitem{yan:outerbd}
Xijin Yan, Jun Yang, and Zhen Zhang.
\newblock An outer bound for multisource multisink network coding with minimum
  cost consideration.
\newblock {\em IEEE/ACM Trans. Netw.}, 14(SI):2373--2385, 2006.

\bibitem{yan:capacity}
Xijin Yan, Raymond~W. Yeung, and Zhen Zhang.
\newblock The capacity region for multi-source multi-sink network coding.
\newblock In {\em Information Theory, 2007. ISIT 2007. IEEE International
  Symposium on}, pages 116 --120, 24-29 2007.

\bibitem{yazdi:capacity1}
S.~M. Sadegh~Tabatabaei Yazdi, Serap~A. Savari, Farzad Farnoud, and Gerhard
  Kramer.
\newblock A multimessage capacity region for undirected ring networks.
\newblock In {\em Information Theory, 2007. ISIT 2007. IEEE International
  Symposium on}, pages 1091 --1095, 24-29 2007.

\bibitem{yazdi:capacity2}
S.M.S. Yazdi, S.A. Savari, K.~Carlson, and G.~Kramer.
\newblock The capacity region of a collection of multicast sessions in an
  undirected ring network.
\newblock In {\em Parallel Processing Workshops, 2007. ICPPW 2007.
  International Conference on}, pages 40 --40, 10-14 2007.

\end{thebibliography}
\bibliographystyle{plain}
\end{singlespace}

\end{document}